\documentclass[12pt]{amsart}

\pdfoutput = 1

\input{macros.sty}

\usepackage{multirow,slashed}

\begin{document}
\onehalfspacing

\title{\large Six-dimensional supermultiplets \\ from bundles on projective spaces}

\author{Fabian Hahner,$^\flat$ Simone Noja,$^\flat$
    Ingmar Saberi,$^\natural$ Johannes Walcher$^\flat$}
\email{fhahner@mathi.uni-heidelberg.de}

\email{noja@mathi.uni-heidelberg.de}

\email{i.saberi@physik.uni-muenchen.de}

\email{walcher@uni-heidelberg.de}

\address{{$^\flat$}Mathematisches Institut der Universit\"at Heidelberg \\ Im Neuenheimer Feld 205 \\ 69120 Heidelberg, Deutschland}

\address{{$^\natural$}Ludwig-Maximilians-Universit\"at M\"unchen \\ Theresienstra\ss{}e 37 \\ 80333 M\"unchen, Deutschland}

\begin{abstract}
    The projective variety of square-zero elements in the six-dimensional minimal supersymmetry algebra is isomorphic to $\P^1 \times \P^3$. We use this fact, together with the pure spinor superfield formalism, to study supermultiplets in six dimensions, starting from vector bundles on projective spaces. We classify all multiplets whose derived invariants for the supertranslation algebra form a line bundle over the nilpotence variety; one can think of such multiplets as being those whose holomorphic twists have rank one over Dolbeault forms on spacetime. In addition, we explicitly construct multiplets associated to natural higher-rank equivariant vector bundles, including the tangent and normal bundles as well as their duals. Among the multiplets constructed are the vector multiplet and hypermultiplet, the family of $\cO(n)$-multiplets, and the supergravity and gravitino multiplets. Along the way, we tackle various theoretical problems within the pure spinor superfield formalism. In particular, we give some general discussion about the relation of the projective nilpotence variety to multiplets 
    and prove general results on short exact sequences and dualities of sheaves in the context of the pure spinor superfield formalism. 
\end{abstract}

\maketitle
\thispagestyle{empty}

\setcounter{tocdepth}{1}
\newpage
\tableofcontents

\section{Introduction}

The study of supermultiplets is an essential ingredient in the construction of interesting supersymmetric field theories. Typically, one characterizes a Lagrangian field theory by specifying its field content, together with the datum of a supersymmetric action functional. The field content consists of some collection of supermultiplets; if a superfield formalism is available, it is typically possible to write down manifestly supersymmetric action functionals in a compact and pleasing form as integrals over superspace.

Over the last fifty years, a variety of methods and techniques have been developed for the construction of supermultiplets and the development of superspace formalisms. 
One such approach is the pure spinor superfield formalism, 
the origins of which date back more than thirty years to the first papers by Nilsson~\cite{Nilsson} and Howe~\cite{HowePS1,HowePS2}.  
The pure spinor approach was developed further, notably in the work of Berkovits on pure spinor methods in worldsheet string theory~\cite[for example]{Berkovits}, and in papers by Cederwall and collaborators for field-theoretic applications. (See~\cite{Cederwall} for a review, as well as references therein.) 
One motivating goal was to provide a formalism suitable for studying supersymmetric higher-derivative corrections and, more generally, classifying possible supersymmetric interactions. From a modern perspective, pure spinor superfields are suited to this purpose because they resolve supermultiplets freely over superspace, leading to simple models of the corresponding interactions: just for example, the action for ten-dimensional super Yang--Mills theory becomes a Chern--Simons action functional, whereas the action for perturbative eleven-dimensional supergravity can be written with only cubic and quartic terms~\cite{Ced-towards,Ced-11d}.

In recent work~\cite{NV,perspectives}, the pure spinor superfield formalism was reinterpreted in modern mathematical language, and shown to work in the general setting of any super Lie algebra $\fp$ of ``super Poincar\'e type.'' (This is equivalent to the condition that $\fp$ admit a consistent $\Z$-grading with support in degrees zero, one, and two; the grading can be viewed as an action of an abelian one-dimensional Lie algebra by endomorphisms, corresponding to rescaling or ``engineering dimension'' from the physics perspective.) Any such algebra has a subalgebra $\ft = \fp_{>0}$ of ``supertranslations,'' and there is a corresponding affine superspace modeled on the group $T = \exp(\ft)$. (For some general context on supertranslation algebras, see~\cite{CCGN}.) The formalism then associates a supermultiplet on the body of~$T$ to any equivariant sheaf on the affine scheme $\hat Y$ of square-zero elements in~$\ft$ that is equivariant for both~$\fp_0$ and rescalings. (In the physics literature, this scheme is often called either the \emph{nilpotence variety}~\cite{NV} or the \emph{space of pure spinors}, although it may not have anything to do with either the spin representation or its minimal orbit, which is the space of pure spinors in the sense of Cartan and Chevalley. With respect to the $\Z$-grading, it is simply the space of Maurer--Cartan elements in~$\ft$.)  In fact, the supermultiplet is naturally freely resolved over~$T$, so that one automatically gets a superfield model.

Set up in this fashion, the pure spinor formalism works to produce multiplets in any dimension and with any amount of supersymmetry; it even works for non-standard examples of supersymmetry, such as the algebras that control residual supersymmetry in twisted theories. In~\cite{spinortwist}, the formalism was shown to commute with twisting in an appropriate sense, and was applied to swiftly and compactly compute the twists of ten- and eleven-dimensional supergravity multiplets, verifying the conjectural description from~\cite{CostelloLi} in the case of type IIB supergravity and providing new results for eleven-dimensional supergravity and type IIA. (See also~\cite{MaxTwist} for a direct computation of the maximal twist of eleven-dimensional supergravity.) 

Given this understanding of the pure spinor superfield formalism, it is, of course, natural to ask whether all multiplets admit a description of this type.
This question was answered in~\cite{EHSequiv}, in which a natural derived generalization of the pure spinor superfield formalism was given. The key observation is that the graded ring $R/I$ of functions on the affine scheme $\hat Y$ is the degree-zero Lie algebra cohomology of~$\ft$. 
As such, it is natural to replace $\hat Y = \Spec R/I$ by the derived model $\Spec C^\bu(\ft) $.
\cite{EHSequiv} showed that the pure spinor superfield formalism generalizes to a functor from equivariant $C^\bullet(\ft)$-modules to supermultiplets for~$\fp$, and in fact provides 
an equivalence of categories. As such, \emph{every} supermultiplet admits a pure spinor superfield description, and arises from a \emph{unique} equivariant sheaf on the affine dg scheme $\Spec C^\bu(\ft)$, up to an appropriate notion of equivalence; the equivariant sheaf can be constructed by taking the derived $\ft$-invariants of the global sections of the multiplet. The equivalence even works when there are no supersymmetries at all. 
As explained in~\cite{EHSequiv}, this result can be thought of as a version of Kapranov's formulation of Koszul duality~\cite{Kapranov}, which relates $D$-modules and $\Omega^\bu$-modules on the same space, performed for the translation-invariant objects on a super Lie group; it is also related to Koszul duality for the graded Lie algebra $\ft$, followed by an associated-bundle construction. (Yet another relation of pure spinors to Koszul duality occurs when viewing the functions on~$\hat Y$ as a commutative algebra and considering the Koszul dual graded Lie algebra, as done in~\cite{MovshevSchwarz,GS,GGST}; further work in this direction will appear in~\cite{CJPS}.)

This result provides further justification for the program of studying $\fp$-multiplets and their interactions using the (derived) algebraic geometry of $C^\bullet(\ft)$-modules. To get the best mileage out of the pure spinor technique, it is natural to start in a setting where the category of $C^\bu(\ft)$-modules, or at least the category of equivariant sheaves on~$\hat Y$, is relatively easy to understand.
A natural candidate is $\N=(1,0)$ supersymmetry in six dimensions, where $\hat Y$ is the space of two-by-four matrices of rank one; as such, the corresponding projective variety $Y = \Proj R/I$ is just $\P^1 \times \P^3$, sitting inside $\P^7$ via the Segre embedding.
In this note we take another step forward in the pure spinor program by providing a detailed case study, as well as developing new methods using techniques from projective algebraic geometry. 
There is a great abundance of geometrically interesting equivariant vector bundles on~$Y$. Taking the direct sum of the global sections of all twists of these bundles, we obtain graded equivariant modules over the ring of functions on the nilpotence variety, so that we can study the associated multiplets. In particular, we classify all multiplets originating from line bundles over $Y = \P^1 \times \P^3$; among others, this recovers the family of so-called ``$\cO(n)$-multiplets'' studied in the literature~\cite{KNTderivative, KuzenkoNovakTheisen, RulesProjective, LinchSugraProj, Galperin86}, and encompasses the vector multiplet and its antifield muliplet, as well as the hypermultiplet. (These three examples have already been studied via the pure spinor superfield formalism in~\cite{CN6d,Ced-6d}.)

Roughly speaking, we provide a link between vector bundles on $Y$ and $\fp$-multiplets in two steps, by combining the connection between quasicoherent sheaves on $\Proj R/I$ and graded $R/I$-modules (which is standard algebraic geometry) with the pure spinor superfield construction. Concretely, we convert a sheaf on $Y$ into a module by forming its graded module of global sections (i.e. by taking the sum of the global sections of all its twists).
Equivariant vector bundles form a subcategory of equivariant quasi-coherent sheaves;
conversely, one can assign a sheaf on $Y$ to each module, though it is important to note that the two operations are not inverses in general.
Following the results on twisting pure spinor superfields in~\cite{spinortwist}, we argue that modules whose associated sheaf is trivial correspond to multiplets that are perturbatively trivial in every twist (see Remark~\ref{rmk:twists}).

In turn, the category of graded equivariant $R/I$-modules sits as a subcategory inside equivariant $C^\bullet(\ft)$-modules, which (by~\cite{EHSequiv}) is equivalent to the category of multiplets.
We can summarize the situation with the following diagram:
\begin{equation}
\begin{tikzcd}
\cat{LineBundles}^{\fp_0}_Y \arrow[r,hook]
& \QCoh_Y^{\fp_0} \arrow[r, shift left, "\Gamma_*"] & \Mod_Y^{\fp_0} \arrow[l,shift left, "\sim"] \arrow[r,hook] & \Mod_{C^\bullet(\ft)}^{\fp_0} \arrow[r,shift left,"A^\bullet"] & \cat{Mult}_{\fp} \arrow[l,shift left, "C^\bullet"] 
\end{tikzcd}
\end{equation}
Since the inverse functor to the pure spinor superfield construction is given by taking derived $\ft$-invariants, classifying all multiplets associated to line bundles thus amounts to classifying all multiplets whose derived $\ft$-invariants are the graded global section module of a line bundle on the projective nilpotence variety. An alternative characterization of such multiplets is via their twists, which are necessarily holomorphic for minimal supersymmetry in six dimensions; in keeping with the results of~\cite{spinortwist}, one expects that the holomorphic twist of such a multiplet is of rank one over Dolbeault forms on the spacetime, and we verify this below. 

In addition to the classification, we translate duality theory for sheaves to the world of multiplets, extending the results of~\cite{perspectives}. To this end, we study the Cohen--Macaulay property and prove that, in good situations, the antifield multiplet can be constructed using the dualizing sheaf on the projective nilpotence variety.

Moving on, we develop some general methods regarding short exact sequences of sheaves in the pure spinor superfield formalism. 
These can be used to tackle higher-rank bundles; generalizations would allow for the construction of the multiplet associated to any higher-rank bundle via a resolution into a chain complex of sums of line bundles, though we do not pursue this in detail here. Our results show that the multiplet associated to a nontrivial extension of two sheaves is a deformation of the direct sum of the multiplets associated to each sheaf by a further differential, and we study such deformations explicitly at the level of component-field presentations of various multiplets.

We then use the Euler exact sequence, as well as the normal and conormal bundle sequences, to explicitly construct the multiplets associated to the tangent bundle, the normal bundle, and their duals. Several of these multiplets are of obvious physical interest; in particular, we identify the supergravity multiplet with the conormal bundle, and the gravitino multiplet with the pullback of the tangent bundle to the ambient space.

\subsection*{Acknowledgements}
We would like to give special thanks to R.~Eager, C.~Elliott, and B.~Williams for numerous conversations and collaboration on closely related projects. We also thank I.~Brunner, M.~Cederwall, J.~Palmkvist, for fruitful conversation.
This work is funded by the Deutsche Forschungsgemeinschaft (DFG, German Research Foundation) under Germany’s Excellence Strategy EXC 2181/1 — 390900948 (the Heidelberg STRUCTURES Excellence Cluster). I.S. is supported by the Free State of Bavaria.

\subsection*{Statement on conflicts of interest}
We certify that each and every author's interest in the material discussed in this manuscript is of
purely scientific nature.  We are not aware of any circumstances that could be construed as a
conflict of interest with any extraneous financial benefits or personal judgement.

\section{Preliminaries}
\label{prelim} 

\subsection{The pure spinor superfield formalism} \label{sec: prelim ps}
We briefly review the pure spinor superfield formalism, as reinterpreted in~\cite{NV,perspectives} and extended in~\cite{EHSequiv}; for a more traditional approach and an overview of the broader literature, the reader is referred to~\cite{Cederwall} and references therein.

\numpar[p:gradings][Grading conventions]
Throughout, we work with cochain complexes, both in the category of graded vector spaces and in the category of super vector spaces. Such an object is a bigraded vector space
\[
\bigoplus_{i,j} V^{i}_{j},
\]
equipped with a differential of bidegree $(i,j) = (1,0)$. Here, $i \in \ZZ$ denotes the cohomological degree; $j$ is viewed as internal to the category of graded vector spaces or super vector spaces. Correspondingly, $j \in \Z$ for graded vector spaces, or $j \in \Z/2\Z$ for super vector spaces. We will refer to the latter as the \emph{internal degree} (for graded vector spaces), or as the \emph{internal parity} (for super vector spaces).
The first grading (arising from the structure of the cochain complex) will be called the \emph{cohomological degree}. The sum of the two will be called the \emph{totalized degree}, or \emph{totalized parity} if understood modulo two.
The Koszul rule of signs applies throughout; the sign is determined by the {totalized} parity, so that $V^{i}_{j}$ is odd precisely when $i+j = 1\pmod 2$. 
The categories we define here were referred to as ``(lifted) dgs vector spaces'' in~\cite[Definition~2.1]{perspectives}, 

In the case when the internal degree is integral, it will sometimes be convenient to use another basis for the bidegree, which (in keeping with~\cite{CJPS}) we will call the \emph{Tate bidegree}. In this basis, we use the totalized degree, together with the negative of the internal degree, which we will refer to as the \emph{weight degree}. Thus $V^i_j$ has Tate bidegree $(i+j,-j)$.
We will use round brackets to refer to shifts of the weight grading, as is typical in projective geometry; we reserve square brackets for shifts in cohomological degree. Note, though, that we use slightly different conventions for each: shifts in the cohomological degree will be indicated explicitly, whereas the weight grading is seen as a $\C^\times$-action and thus left implicit. So we will write, for example, 
\[
V = \bigoplus_{i,j} V^{i}_{j}[-i].
\]

One advantage of the Tate bidegree is that it places the graded rings that appear in our constructions in totalized degree zero (so that constructions such as the spectrum work in standard fashion), with the weight grading assigning generators weight one as is standard. 
We summarize our conventions for the various generators of the commutative differential graded algebras we will construct in the sequel in Table~\ref{tab:grad}. The notation and other definitions will be explained in the sequel; the table and the following remarks will thus primarily be useful for future reference.

\begin{table}
    \begin{tabular}{cc|cc|cc|c}
        & & \multicolumn{2}{|c|}{\emph{Standard bidegree}} & \multicolumn{2}{|c|}{\emph{Tate bidegree}} \\
        & & cohomological & internal & totalized & weight & parity \\ 
        \hline
        \multirow{2}{*}{$C^\bu(\ft):$} & $\lambda \in \ft_1^\vee[-1]$ & $1$ & $-1$ & $0$ & $1$ & $+$ \\
        & $v \in \ft_2^\vee[-1]$ & $1$ & $-2$ & $-1$ & $2$ & $-$ \\ \hline 
        \multirow{2}{*}{$C^\infty(T):$} & $\theta \in \ft_1^\vee$ & $0$ & $-1$ & $-1$ & $1$ & $-$ \\
        & $x \in C^\infty(T_+) $ & $0$ & $0$ & $0$ & $0$ & $+$ \\ \hline
        \multirow{2}{*}{shifts:} & $\C[-1]$ & $1$ & $0$ & $1$ & $0$ & $-$  \\
        & $\C(-1)$ & $1$ & $-1$ & $0$ & $1$  & $+$
    \end{tabular}
    \caption{Grading conventions for generators of~$A^\bu$}
    \label{tab:grad}
\end{table}

When working with the component-field presentations of multiplets, we will use yet another basis for the bidegree, consisting of the cohomological and the totalized degrees taken together. In this basis, parity is determined by totalized degree modulo two. The advantage of this basis is that the bidegree corresponds (up to sign) to powers of the two generators $\lambda$ and $\theta$. The differential on a multiplet is of cohomological degree $+1$, but is not of homogeneous internal degree. Rather, the terms in the differential that are differential operators of degree $k$ are of totalized degree $1 - 2k$. (This occurs because $x$ is placed in degree zero, rather than $-2$, so as to obtain a bigraded vector bundle rather than just a bigrading on sections over the formal disk.)
Multiplets thus inherit a filtration from the \emph{internal} degree that is respected by the differential. 

We will often present supermultiplets using a tabular notation, with the cohomological display on the vertical axis (increasing downwards) and the totalized degree on the horizontal axis (decreasing to the right; note that the totalized degree of a multiplet is typically nonpositive, and counts \emph{negative} powers of the generators $\theta$).

\numpar[alg][Supertranslation algebras]
Let $\fp$ be a graded Lie algebra, supported in internal degrees zero, one, and two, and let $\ft = \fp_{>0}$ be its positively graded subalgebra. Concretely, this means that $\ft_1$ and $\ft_2$ carry representations of~$\fp_0$, and that the bracket $\Sym^2(\ft_1) \to \ft_2$ is $\fp_0$-equivariant. The central example will be any super Poincar\'e algebra, for which $\ft_1$ consists of supersymmetries, $\ft_2$ of translations, and $\fp_0$ of Lorentz and $R$-symmetries. 
In general, $\fp_0$ consists of (a sub Lie algebra of) the derivations of~$\ft$; all our constructions will be equivariant with respect to such endomorphisms. We emphasize that $\fp$ is concentrated in cohomological degree zero. Reducing the internal degree modulo two allows us to view $\fp$ as a super Lie algebra. 

We consider the Lie algebra cochains of~$\ft$, which are a bigraded commutative differential algebra with differential of bidegree $(1,0)$, freely generated by even elements in bidegree $(1,-1)$ and odd elements in bidegree $(1,-2)$.
Since $\ft$ is a central extension of the form 
\deq{
    0 \to \ft_2 \to \ft \to \ft_1 \to 0,
}
there is a corresponding fiber sequence of augmented bigraded cdga's of the form\footnote{Notice that in the category of \emph{augmented} bigraded cdga's $\mathbb{C}$ is the zero object.}
\deq{ \label{aucdga}
    \C \to C^\bu(\ft_1) \to C^\bu(\ft) \to C^\bu(\ft_2) \to \C.
}
The first arrow witnesses $C^\bu(\ft)$ as an algebra over the graded polynomial ring 
\deq{
    R = C^\bu(\ft_1) = \Sym^\bullet(\ft_1^\vee[-1]), 
}
whose generators sit in bidegree $(1,-1)$ and are thus of even parity. In fact, with respect to the Tate bidegree, $C^\bu(\ft)$ is the Koszul complex~\cite{Eisenbud, NoRe} over $R$ defined by the $\ft_2$-valued set of equations
\begin{equation}
    [Q,Q] = 0,
\end{equation}
where $Q \in \ft_1$.

These equations define a homogenous ideal $I\subset R$; the space 
\deq{
    \hat{Y} = \Spec R/I \subset \Spec R = \ft_1[1]
}
associated to this ideal is called the nilpotence variety, and consists of the space of square-zero supercharges in~$\ft$. (We emphasize, though, that we think of $\hat{Y}$ as an affine scheme.)
This space has been studied extensively for physical supertranslation algebras in~\cite{NV, ChrisPavel}. 
Explicitly, we can expand $Q$ in a basis $Q = \lambda^\alpha Q_\alpha$ and identify $R = \C[\lambda^\alpha]$. Denoting the structure constants of the bracket by $\Gamma^\mu_{\alpha \beta}$, the ideal then takes the form
\begin{equation}
I = (\lambda^\alpha \Gamma^\mu_{\alpha \beta} \lambda^\beta).
\end{equation}
We further observe that, with respect to the totalized grading, 
$H^0(\ft) = R/I$. In general, though, $I$ is not a complete intersection, so that the Koszul complex of~$I$ does not define a resolution of~$R/I$. Rather, $C^\bu(\ft)$ should be thought of as a derived enhancement of the standard (affine) nilpotence variety $\hat{Y}$. 

\numpar[formalism][Constructing multiplets] \label{ssec: construct}
As recalled in the introduction, the pure spinor superfield formalism is an equivalence of categories between $\fp$-multiplets and $\fp_0$-equivariant $C^\bu(\ft)$-modules; it restricts to $R/I$-modules, giving a systematic way to construct a class of supersymmetric multiplets by studying sheaves on the affine scheme $\hat{Y}$ of square-zero supersymmetries. 
Here, we give a rough overview, reviewing the aspects which are relevant to this work; for details, we refer to~\cite{perspectives,EHSequiv}. In particular, we make use of the notion of a supermultiplet as defined in~\cite{perspectives}. A supermultiplet is a chain complex of affine super vector bundles\footnote{For $V = \mathbb{R}^d$ or $V = \mathbb{C}^d$, a vector bundle $E$ on $V$ is called affine if its total space carries an action of the affine group $\mbox{Aff} (V) = V \rtimes \mbox{Spin} (V)$ such that the structural projection $E \stackrel{\pi}{\longrightarrow} V$ is equivariant with respect to the action of the affine group on the base $V$. See \cite[\S2.4.3]{perspectives} for details.} on spacetime (which here is a torsor for the vector space $\mathfrak{t}_2$ of translations), equipped with a homotopy action of the super Poincar\'e algebra that is compatible with the affine structure.

\begin{rmk}
    \label{rmk:pert}
    The notion of a supermultiplet, as defined here, is appropriate for models of \emph{perturbative} field theory.
    This is reflected in the fact that all ``fields'' are sections of vector bundles, rather than (for example) maps into more general manifolds. As such, every interacting field theory constructed in this manner will be a \emph{formal} moduli problem. We will not have need to discuss such issues in greater detail here, but refer the interested reader to~\cite{CG1,CG2}. We will, however, occasionally use the term ``perturbative'' for emphasis in the sequel, when we find it necessary to emphasize the context of our results.
\end{rmk}

We denote the supergroup of supertranslations by $T = \exp(\ft)$; its body is a vector space $\R^d$, thought of as an additive abelian group. There is a corresponding supergroup $P$, the analogue of the super Poincar\'e group, which is the semidirect product of $T$ with a group $P_0$ of its outer automorphisms. Recall that the algebra of \emph{free superfields} is defined to be the supercommutative algebra of smooth functions on~$T$: 
\begin{equation}
    C^\infty(T) = C^\infty(\R^d) \otimes \wedge^\bullet \ft_1^\vee = C^\infty(\R^d) \otimes_\C \C[\theta^\alpha] .
\end{equation}
(For a recent source for the relevant supergeometry, we refer to~\cite{Noja2021geometry}.) 
We will write $x^\mu$ for coordinates on the even part of~$T$ and $\theta^\alpha$ for anticommuting coordinates on the odd part. 
$C^\infty(T)$ is the regular representation of~$T$, and thus acquires two distinct actions of $\ft$ by derivations, namely by the left and right actions of~$T$ on itself. Each extends to an action of~$\fp$. We denote these by 
\begin{equation}
\mathpzc{R}\, ,\mathpzc{L} 
: \fp \longrightarrow \mathrm{Vect}(T) \: .
\end{equation}

Now let $\Gamma$ be a graded equivariant $R/I$-module\footnote{Here and in the following, equivariance will always mean $\mathfrak{p}_0$-equivariance, unless differently specified. Further, each weight-graded piece of $\Gamma$ is assumed to be $\mathfrak{p}_0$-equivariant.}. The pure spinor superfield associated to $\Gamma$ is the cochain complex
\begin{equation}
    A^\bullet(\Gamma) = \left(C^\infty(T) \otimes \Gamma \: , \: \cD = \lambda^\alpha \mathpzc{R}(Q_\alpha) \vphantom{3^3} \right) \: .
\end{equation}
The super Poincar\'e algebra acts on $A^\bullet(\Gamma)$ on the left, endowing $A^\bullet(\Gamma)$ with the structure of a strict multiplet~\cite[\S3.2 (for details) and \S3.5 (for an example)]{perspectives}. 

$A^\bullet(\Gamma)$ is equipped with a bigrading as explained above in~\S\ref{p:gradings}. 
The differential is odd and has cohomological degree $+1$, but is not homogeneous for the second integer grading. It does, however, preserve the decreasing filtration associated to the internal degree, which we will call the \emph{standard filtration}. Terms in the differential that are spacetime differential operators of order $k$ have internal degree $-2k$, and thus totalized degree $1-2k$. The super vector bundle structure arises by taking the totalized degree modulo two.

While $A^\bullet(\Gamma)$ is a strict multiplet, it is typically very large and does not resemble the standard component-field descriptions. These component-field descriptions arise by identifying a quasi-isomorphic subcomplex of~$A^\bullet(\Gamma)$, choosing a retraction to this complex, and transferring all the relevant structures along this retraction via homotopy transfer. 

A minimal component-field description can be canonically constructed as follows. In coordinates, the differential $\cD$ takes the form
\begin{equation}
    \cD = \cD_0 + \cD_1 = \lambda^\alpha \frac{\partial}{\partial \theta^\alpha} - \lambda^\alpha \Gamma^\mu_{\alpha \beta} \theta^\beta \frac{\partial}{\partial x^\mu},
\end{equation}
where $\cD_i$ has internal degree $-2i$.
It is clear upon inspection that
\begin{equation}
    \Gr A^\bu(\Gamma)
    = C^\infty(\R^d) \otimes \left( K^\bullet(\Gamma)   , \: \d_K \vphantom{3^3} \right)  ,
\end{equation}
where $(K^\bullet(\Gamma) , \d_K)$ is the Koszul complex of $\Gamma$. 
(We identify the bigrading on the Koszul complex with the Tate bidegree.) 
This implies that the fields of the minimal multiplet take values in the Koszul homology of $\Gamma$. 
(Note that, in our conventions, $K^\bu$ is cohomologically, therefore nonpositively, graded.)

On general grounds, Koszul homology can be computed by considering a minimal equivariant free resolution $(L,d_L)$ of $\Gamma$ in $R$-modules (for more details, as well as the existence and uniqueness of the component fields, we refer to~\cite[\S3.4.1 and Lemma~3.10]{EHSequiv}). 
Such a minimal free resolution is also bigraded, by resolution degree and by the weight grading on~$R$. In our conventions, this bigrading is identified with the Tate bidegree. 
$L^\bu$ sits in nonpositive totalized degree, and the differential $d_L$ has Tate bidegree $(1,0)$. 
For convenience, we will often write this in the form 
\deq{ \label{freeres}
    L^\bu = \bigoplus_{k \geq 0} L^{-k} = \bigoplus_{k \geq 0} W^k\otimes R[k],
}
where $W^k$ is the finite-dimensional weighted $\fp_0$-representation in which the generators of~$L^{-k}$ transform. 
There is then an isomorphism of Tate-bigraded $\fp_0$-modules between the generators of $L^\bu$ and the Koszul homology. 
This was established in~\cite{MovshevSchwarzXu, LosevCS}, and is a key insight,
allowing one to quickly compute component fields by freely resolving $\Gamma$ over~$R$.
In~\cite[\S\S3.6--7]{perspectives} and~\cite[Theorem 4.8ff]{EHSequiv}, it was further shown that the resolution differential $d_L$ encodes the $\fp$-module structure (as had been conjectured in~\cite{BerkovitsSupermembrane}).

It remains to recall how the structure of a multiplet is defined on the component fields. Choosing an equivariant set of representatives for Koszul homology, we consider a deformation retract of the form 
\begin{equation}
\begin{tikzcd}
    \arrow[loop left]{l}{h}\left( \vphantom{3^3} \Gr A^\bullet (\Gamma)  , \: \cD_0 \right) \arrow[r, shift left, "p"] & \left( \vphantom{3^3} H^\bullet(\Gr A^\bullet(\Gamma))  , \, 0 \right) \arrow[l, shift left, "i"] .
\end{tikzcd}
\end{equation}
Applying homotopy transfer to this data, 
we obtain a new multiplet with underlying complex $H^\bullet(\Gr A^\bullet(\Gamma))$, new differential $\cD'$ (obtained from $\cD_1$ via homotopy transfer), and supersymmetry action $\rho'$ (obtained from $L$).\footnote{See~\cite[\S 3.6]{perspectives} for more details; in particular, the associated spectral sequence always degenerates for degree reasons, since the differential on the $E_n$ page is of degree $2n-1$ in the $\theta$ variables.} The multiplet obtained in that way is minimal in the sense that its differential only contains differential operators of non-zero degree; we emphasize, however, that this characterization is model-dependent. By construction, this multiplet is quasi-isomorphic to $A^\bullet(\Gamma)$; in physics language, a quasi-isomorphism in the category of multiplets encodes the notion of an equivalence of perturbative free field theories. (See~\cite[\S 2]{EHSequiv} for more discussion). We call this multiplet the component-field multiplet (or minimal multiplet) associated to $\Gamma$, and denote it by $\mu A^\bullet(\Gamma)$. When we want to refer to the underlying vector bundle of $\mu A^\bullet(\Gamma)$---which is the associated bundle of the bigraded Lorentz representation on $H^\bu(K^\bu(\Gamma))$---we will write $\mu A^\bullet(\Gamma)^\#$. This notation will in general apply to any multiplet and denote its underlying bigraded vector bundle, considered without the data of the differential and the module structure.

\numpar[hilbser][Computational techniques]
In practice, there are various ways to extract information on $\mu A^\bu(\Gamma)$ from the algebraic input of a $\ZZ$-graded module $\Gamma$, always assumed to be finitely generated. The Hilbert series of such a module $\Gamma$ is its graded dimension, viewed as a formal Laurent series 
\begin{equation} \label{fls}
    \grdim(\Gamma) = \sum_{\ell \in \ZZ}\dim(\Gamma_\ell) \: t^\ell \in \Z(\!(t)\!),
\end{equation}
where the $\Gamma_\ell$ are the $\ell$-weighted summands of $\Gamma$. Note that the series is well-defined, as the graded components $\Gamma_\ell$ are finite-dimensional vector spaces.
Taking $\ell_0$ to be the $t$-adic valuation of $\grdim(\Gamma)$, we can write
\begin{equation} \label{fps}
    t^{-\ell_0} \cdot \grdim(\Gamma) =  \sum_{\ell = 0}^\infty \dim(\Gamma_{\ell+\ell_0}) \: t^\ell \in \Z[\![t]\!].
\end{equation}
This implies that, up to a total shift of $\Gamma$ by $\ell_0$, one gets a formal power series, instead of a formal Laurent series; this shift corresponds to a shift in the cohomological degree of the corresponding multiplet.
By the Hilbert syzygy theorem, this series can be resummed to yield
\begin{equation} \label{grdimP}
\grdim(\Gamma) = t^{\ell_0}\frac{P(t)}{(1-t)^n},
\end{equation}
where $n= \dim(\ft_1)$ is the number of variables in $R$ and $P$ is a polynomial with integer coefficients.
The coefficients of $P$ then describe the graded free $R$-module underlying the minimal resolution $(L,d_L)$, and thus record the Betti numbers of the graded bundle $\mu A^\bu(\Gamma)^\#$. Concretely, it follows from~\eqref{freeres} that
\deq{
    \begin{aligned} 
    \grdim(\Gamma) = t^{\ell_0}\sum_{k=0}^{\infty} (-1)^k \grdim(L^{-k}) &= t^{\ell_0}\sum_{k=0}^{\infty} (-1)^k \grdim(W^k) \grdim(R)  \\
    &= \frac{t^{\ell_0}}{(1-t)^n} \sum_{k=0}^{\infty} (-1)^k \grdim(W^k).
\end{aligned}
}
The fact that $L^\bu$ is minimal means that all terms in the differential are in the maximal ideal of~$R$, which implies that $\wt(W^k) \geq k$. Thus $t^k \mid \grdim(W^k)$.
Taking \eqref{grdimP} into account, we deduce that
\deq{
    P(t) = \sum_{k=0}^{\infty} (-1)^k \grdim(W^k) = \sum_{k,\ell} (-1)^k \dim (W^k_\ell) t^\ell, 
}
recalling that the pair $(k,\ell)$ refers to Tate bidegree.
Of course, cancellations are possible, so that $P(t)$ does not contain complete information about the Koszul homology. 
But in practice, one is often in a situation where no cancellations occur. In any case, if an acyclic resolution differential can be constructed, one has demonstrated that no further generators appear in~$L^\bu$.

This procedure can be improved in the $\fp_0$-equivariant setting. Since $\Gamma$ is graded and equivariant, each weighted summand $\Gamma_\ell$ is a finite-dimensional representation of $\fp_0$. 
We can consider the \emph{equivariant} Hilbert series as a formal power series in the representation ring of $\fp_0$\footnote{The representation ring of $\fp_0$ is the free abelian group on the set of finite-dimensional irreducible $\fp_0$-representations, with the multiplication induced by the tensor product of representations.} by setting
\begin{equation} \label{hilb}
    \Hilb(\Gamma) = t^{\ell_0} \sum_{\ell=0}^{\infty} \Gamma_\ell \: t^\ell  \in \Rep(\fp_0)(\!(t)\!).
\end{equation}
We can then rewrite the Hilbert series in the form
\begin{equation}
\begin{split}
    \Hilb(\Gamma) &= t^{\ell_0} \left[ \sum_{d = 0}^{\infty} \Sym^d(\ft_1^\vee) \: t^d \right] \otimes \left[ \sum_{k,\ell} (-1)^k W_\ell^k  t^\ell \right]  \\
    &= t^{\ell_0} \Hilb(R) \cdot \Hilb(\chi(W^\bu)),
\end{split}
\end{equation}
where we write $\chi (W^\bullet)$ for the element $\sum_k (-1)^k W^k$ of the representation ring.
Comparing coefficients order by order, one obtains a system of equations which allows to identify $\chi(W^\bu)$, and thus (at least in favorable cases) $W^\bu$ itself:
\begin{equation} \label{eq: recursion gen}
\begin{split}
    \chi(W^\bu)_0 &= \Gamma_{\ell_0} \\
    \chi(W^\bu)_1 &= \Gamma_{\ell_0 +1} - \ft_1^\vee \otimes \chi(W^\bu)_0 \\
\vdots& \\
\chi(W^\bu)_k &= \Gamma_{\ell_0 + k} - \sum_{d=1}^{k} \Sym^d(\ft_1^\vee) \otimes \chi(W^\bu)_{k-d}.
\end{split}
\end{equation} 
This technique is used frequently in the work of Cederwall and collaborators, and we will apply it in examples in what follows.

\subsection{Supersymmetry in six dimensions}
We review the concrete form of the above constructions for the six-dimensional $\N=(1,0)$ supertranslation algebra. 
We always work with complexified algebras and thus ignore signature. The translations are thus $\ft_2 = V = \C^6$, where $V$ denotes the vector representation of $\Spin(6)$.

We denote the two chiral spinor representations of $\Spin(6)$ by~$S_\pm$. There is an exceptional isomorphism identifying $\lie{so}(6) \cong \lie{sl}(4)$.
Under this isomorphism, $S_+$ is identified with the fundamental representation of $\lie{sl}(4)$ and $S_- = (S_+)^\vee$ with the antifundamental, while $V$ is identified with the two-form.
There are thus $\Spin(6)$-equivariant isomorphisms
\begin{equation} \label{eq: wedge iso}
\wedge^2 (S_\pm) \cong V .
\end{equation}
The six-dimensional $\cN=(1,0)$ supertranslation algebra takes the form
\begin{equation}
    \ft = V(-2) \oplus ( S_+ \otimes U)(-1),
\end{equation}
where $U= (\C^2,\omega)$ is a symplectic vector space. 
The bracket takes the form
\begin{equation} \label{eq: bracket}
[-,-] = \wedge \otimes \omega : (S_+ \otimes U) \otimes (S_+ \otimes U) \longrightarrow V  ,
\end{equation}
where the isomorphism~\eqref{eq: wedge iso} is used. 
The bracket is thus equivariant with respect to the obvious action of $\Sp(1)$ on $U$, identifying the $R$-symmetry group as $\Sp(1) \cong SL(2)$, since we are working in a complexified setting. (The relevant Lie algebra is thus $\mathfrak{sl}(2)$.)

Explicitly, using a basis $Q_\alpha^i$ we can write
\begin{equation}
[Q_\alpha^i, Q_\beta^j] = \varepsilon^{ij} \Gamma^\mu_{\alpha \beta} P_\mu \: .
\end{equation}
Here $i,j=1,2$ are indices for $U$ and $\alpha,\beta=1 \dots 4$ are indices for $S_+$. Identifying $R = \C[\lambda^\alpha_i]$, the defining ideal of the nilpotence variety is spanned by the six quadratic equations
\begin{equation}
I = (\lambda^\alpha_i \Gamma^\mu_{\alpha\beta} \varepsilon^{ij} \lambda^\beta_j) \: .
\end{equation} 
These equations can be conveniently packaged as follows. Let us define the rank of a supercharge $Q \in \ft_1$ to be the rank of the associated linear map $(S_+)^\vee \longrightarrow U$. From~\eqref{eq: bracket} it follows immediately that the square-zero supercharges are precisely those of rank one. In terms of coordinates this means that the ideal $I$ is spanned by the $2\times 2$ minors of the matrix with entries $\lambda^\alpha_i$,
\begin{equation}
\begin{pmatrix}
\lambda^1_1 & \lambda^2_1 & \lambda^3_1 & \lambda^4_1 \\
\lambda^1_2 & \lambda^2_2 & \lambda^3_2 & \lambda^4_2
\end{pmatrix} .
\end{equation}
Accordingly, the nilpotence variety $\hat{Y} = \Spec R/I$ can be thought of as the space of rank one matrices inside $M^{2\times 4}(\C)$. Its projective version $Y =\Proj R/I$ can be identified with the product of two projective spaces via the Segre embedding. In more detail, the square-zero supercharges are precisely those which can be written as
\begin{equation}
Q = \xi \otimes r \quad \text{with} \quad \xi \in S_+ \: , \: r \in U .
\end{equation}
Interpreting $[r_0:r_1]$ and $[\xi_0:\dots :\xi_3]$ as homogeneous coordinates on $\P^1$ and $\P^3$ respectively identifies $Y$ with the image of the Segre embedding
\begin{equation}
\sigma: \P^1 \times \P^3 \longrightarrow \P^7 \qquad ([r_0:r_1] , [\xi_0:\dots \xi_3]) \mapsto [r_0 \xi_0: \dots : r_1 \xi_3] \: .
\end{equation}
We can thus explore supermultiplets in six dimensions using the algebraic geometry of projective spaces.
\subsection{From sheaves on projective schemes to modules}
As explained above, the pure spinor superfield formalism constructs a supersymmetric multiplet from a graded equivariant $R/I$-module. Clearly, these are closely related to sheaves of $\cO_{\hat{Y}}$-modules on $\hat{Y}$: For any affine scheme $X = \Spec S$ there is an equivalence of categories between quasi-coherent sheaves of $\cO_X$-modules and $S$-modules. Explicitly, this equivalence is given by taking global sections
\begin{equation}
\QCoh_{\cO_X} \longrightarrow \Mod_S \qquad  \cF \mapsto \Gamma(\mathrm{Spec}(S), \cF) \: ,
\end{equation}
and conversely assigning
\begin{equation}
\Mod_S \longrightarrow \QCoh_{\cO_X} \qquad  M \mapsto \tilde{M} \: ,
\end{equation}
where $\tilde{M}$ is defined by the requirement $\tilde{M}(D_f) = M_f$ for all $f \in S$.\footnote{Here $D_f \subseteq \Spec S$ denotes all prime ideals of $S$ not containing $f$ and $M_f$ the localization of $M$ at $f$.} If $S$ is graded, one can think of the grading as defining a $\lie{gl}_1$-action on~$\Spec S$; it is then possible to define an equivalence between graded $S$-modules and quasicoherent sheaves of~$\O_X$-modules on~$X$ that are equivariant for rescalings.

One can thus always think of the input to the (underived) pure spinor superfield formalism geometrically as a $(\fp_0 \oplus \lie{gl}_1)$-equivariant sheaf on the affine nilpotence variety. 
It is tempting to ask if one can picture the situation using the geometry of sheaves on~$Y = \Proj R/I$.
Here, the situation is geometrically compelling, but a bit less unequivocal. From a graded $S$-module $M$, we can construct a quasi-coherent sheaf on $\Proj S$ by setting $\tilde{M}(D_f) = (M_f)_0$. By the definition of the $\Proj$-construction we have $\tilde{S} = \cO_{\Proj S}$. The twisting sheaves are defined by
\begin{equation}
\cO_{\Proj S}(n) = \widetilde{S(n)} \: .
\end{equation}
For a sheaf $\cF$ on $\Proj S$, we define the associated (graded) $S$-module to be
\begin{equation} \label{Gammast}
\Gamma_*(\cF) = \bigoplus_{n \in \Z} H^0(\Proj S , \cF(n)) \: ,
\end{equation}
where $H^0 = \Gamma : \QCoh_{\cO_X} \rightarrow \Mod_{\CC}$ is the global section functor. We will call $\Gamma_*(\cF)$ the graded global section module of $\cF$. The relation between Serre twists and internal grading of $\Gamma_\ast (\cF)$ is given by 
$\Gamma_*(\cF)_n = H^0 (\Proj S, \cF(n)).$
With a slight abuse of notation, given an equivariant sheaf $\cF$ on $\Proj S$ we will define its equivariant Hilbert series as 
\begin{equation}
\Hilb({\mathcal{F}}) \defeq \Hilb (\Gamma_\ast (\mathcal{F})),
\end{equation} 
where $ \Hilb(\Gamma_\ast (\mathcal{F}))$ is defined as in \eqref{hilb}.

Finally, it is to be observed that in general the above assignments \eqref{Gammast} no longer give an equivalence of categories, but we can still use $\Gamma_*(-)$ to construct large families of input data for the pure spinor superfield formalism from sheaves on the projective version of the nilpotence variety. This is in particular useful in the case of $\cN=(1,0)$ supersymmetry in six dimensions, since---as we explained above---the projective version of the nilpotence variety can be identified with $\P^1 \times \P^3$ and equivariant sheaves on this space are very well understood geometrically.

\subsubsection{What the projective perspective misses}
Contrary to the affine case, the functors $\sim$ and $\Gamma_*$ do not yield an equivalence of categories. While it is true that
\begin{equation}
	\widetilde{\Gamma_*(\cF)} \cong \cF
\end{equation}
for any quasicoherent sheaf $\cF$, it can happen that $\Gamma_*(\tilde{M})$ is not isomorphic to the original module $M$. Let us restrict to the case where $S = R$ is a polynomial ring and $M$ is a finitely generated graded module. 
Consider the class $\cC$ of modules $M$ such that $M_n = 0$ for $n$ large enough. One finds that these are precisely the modules which are in the kernel of $\sim$. One has the following result:
\begin{prop}[\cite{Serre1955}]
	Let $M$ be a graded $S$-module. Then
	\begin{equation}
		\tilde{M} = 0 \iff M \in \cC \: .
	\end{equation}
\end{prop}
For the pure spinor superfield formalism, this means that multiplets corresponding to modules which are concentrated in finitely many degrees cannot be obtained from sheaves on the projective nilpotence variety. One such example is the free superfield $A^\bullet(\C)$ which is constructed from the trivial module $\C$ (thought of as the quotient of $R$ by the maximal ideal corresponding to the origin). The corresponding sheaf on the affine nilpotence variety is the skyscraper sheaf with value $\C$ at the origin; the associated sheaf on the projective nilpotence variety is trivial.

In general, such sheaves must have zero-dimensional support. The support of an equivariant sheaf must consist of a union of orbits of the $P_0$-action; since we only consider sheaves that are equivariant for rescaling, the origin is the unique zero-dimensional orbit, so that any module in the kernel of~$\sim$ defines a sheaf supported entirely at the origin. 

\begin{rmk}
    \label{rmk:twists}
    It is natural to wonder how conditions on the support of a sheaf translate into properties of the corresponding multiplet. An intuitive answer is suggested by the results of~\cite{spinortwist} on twisting in the pure spinor formalism. There, it was noted that deforming a super Poincar\'e-type algebra by a square-zero supercharge commutes with forming the pure spinor multiplet of the structure sheaf. When $Y$ is smooth (as is the case here), only holomorphic twists are available, and the computations in~\cite{spinortwist} imply that the holomorphic twist of a given multiplet is freely generated over the Dolbeault complex on spacetime by the stalk of the corresponding sheaf at the holomorphic supercharge. We do not explain this in detail here, but will remark from time to time on the physical interpretations of our results that it suggests.
\end{rmk}

In keeping with Remark~\ref{rmk:twists}, 
we expect that multiplets corresponding to sheaves in the kernel of~$\sim$ are precisely those that are perturbatively trivial in every possible twist. We note that the free superfield falls into this class.

\subsection{Some natural equivariant vector bundles} 
In the bulk of this work we are going to consider various vector bundles over the nilpotence variety $Y  \cong \P^1 \times \P^3$ and construct the associated multiplets using the pure spinor superfield formalism. For later reference and completeness, we now introduce the bundles that will appear later on.
\subsubsection{Line bundles}
The product space geometry of the nilpotence variety $Y \cong \proj 1 \times \proj 3$ makes it easy to describe all of its line bundles. Indeed, holomorphic line bundles are classified up to isomorphism by the Picard group $\Pic (Y) \cong H^1 (Y, \mathcal{O}^\ast_Y)$, which can be easily computed using the exponential short exact sequence
\bear
\xymatrix{
	0 \ar[r] & \mathbb{Z}_{\proj 1 \times \proj 3} \ar[r] & \mathcal{O}_{\proj 1 \times \proj 3}\ar[r] & \mathcal{O}_{\proj 1 \times \proj 3}^\ast \ar[r] & 0
}
\eear
and its related long exact sequence in cohomology. In particular, one finds the isomorphism $\Pic (\proj 1 \times \proj 3) \cong \mathbb{Z}\oplus \mathbb{Z}$, which tells that every line bundle on the product variety $\proj 1 \times \proj 3$ arises from line bundles defined on its factors $\proj 1$ and $\proj 3$. (Recall that $\Pic (\proj n) \cong \mathbb{Z}$ for any $n\geq 1 $.) In fact, given the structural projections 
\bear
\xymatrix{
	& \ar[dl]_{\pi_1} \proj 1 \times \proj 3 \ar[dr]^{\pi_3}\\
	\proj 1  & & \proj 3}  
\eear
from $\mathbb{P}^1 \times \mathbb{P}^3$ to its cartesian components, the line bundles on $\proj 1 \times \proj 3$ are all given by the exterior tensor product of a pair line bundles defined over $\proj 1$ and $\proj 3$ respectively. In other words,
\deq{
\mathcal{O}_{\proj 1 \times \proj 3} (n, m) \defeq 
    \mathcal{O}_{\proj 1} (n) \boxtimes \mathcal{O}_{\proj 3} (m) \defeq \pi^\ast_1 \mathcal{O}_{\proj 1} (n) \otimes_{\mathcal{O}_{ \mathbb{P}^{\scalemath{0.4}{1}} \times \mathbb{P}^{\scalemath{0.4}{3}} } }\pi^\ast_3 \mathcal{O}_{\proj 3} (m)  
\: ,
\quad (n, m) \in \mathbb{Z}^{\oplus 2}.
}
(The notation is standard.) Note that the generators of the Picard group $\Pic (\mathbb{P}^1 \times \mathbb{P}^3)$ are given by $\mathcal{O}_{\proj 1 \times \proj 3}(1, 0)$ and $\mathcal{O}_{\proj 1 \times \proj 3}(0, 1)$; the connecting (iso)morphism $\delta_2 : \Pic (\mathbb{P}^1 \times \mathbb{P}^3) \rightarrow \mathbb{Z}\oplus\mathbb{Z}$ thus carries $\cO_{\proj 1 \times \proj 3}(n,m)$ to $(n,m)$, and tensor product yields an isomorphism (of $\mathcal{O}_{\proj 1 \times \proj 3}$-modules) 
\bear
\mathcal{O}_{\proj 1 \times \proj 3} (n, m) \otimes \mathcal{O}_{\proj 1 \times \proj 3} (k , l) \cong \mathcal{O}_{\proj 1 \times \proj 3} (n + k, m + l)
\eear 
for any $(n, m), (k, l) \in \mathbb{Z}^{\oplus 2}$. We will often use the shorthand $\O(n,m) = \O_{\proj 1 \times \proj 3}(n,m)$, since we focus on the example of six-dimensional minimal supersymmetry in this paper. Finally, we will denote a $k$-twisting sheaf for the nilpotence variety $Y$ by 
\bear
\mathcal{O}_Y (k) \defeq \mathcal{O}_{\proj 1 \times \proj 3} (k, k). 
\eear

\subsubsection{Tangent and cotangent bundles}
Similarly, tangent and cotangent bundles on a product variety can be reconstructed by the tangent and cotangent bundles of its Cartesian components. In fact, the tangent bundle of $\proj 1 \times \proj 3 $ is given by the exterior direct sum
\bear \label{tang}
\mathcal{T}_{\proj 1 \times \proj 3} \cong  \pi^\ast_{1} \mathcal{T}_{\proj 1} \oplus \pi^\ast_3 \mathcal{T}_{\proj 3}  \defeqinv \mathcal{T}_{\proj 1} \boxplus \mathcal{T}_{\proj 3}. 
\eear
Note that $\mathcal{T}_{\proj 1} $ is a line bundle and one has $\mathcal{T}_{\proj 1} \cong \mathcal{O}_{\proj 1} (+2)$, while $\mathcal{T}_{\proj 3}$ is an ample non-decomposable vector bundle of rank three. The tangent bundle on any projective space $\P^n$ sits in the Euler exact sequence
\bear \label{euler}
\xymatrix{
	0 \ar[r] & \mathcal{O}_{\proj n} \ar[r] & \mathcal{O}_{\proj n} (+1) \otimes V_{n+1} \ar[r] & \mathcal{T}_{\proj n} \ar[r] & 0,
}
\eear
where $V_{n+1}$ is a $(n+1)$-dimensional complex vector space that carries the fundamental representation of $\mathfrak{sl}_{n+1}$. The Euler exact sequence \eqref{euler} is a short exact sequence of $\mathfrak{sl}_{n+1}$-equivariant sheaves; this will play a role in \S\ref{sec: normal}, when we will study the multiplet associated to the tangent bundle $\mathcal{T}_Y $ of the nilpotence variety.

In a similar fashion, the cotangent bundle $\Omega^1_{Y} \defeq \SHom_{\mathcal{O}_{\mathbb{P}^{\scalebox{0.3}{1}} \times \mathbb{P}^{\scalebox{0.3}{3}}}} (\mathcal{T}_{\proj 1 \times \proj 3} , \mathcal{O}_{\proj 1 \times \proj 3})$ of the nilpotence variety $Y$ is given by the exterior direct sum 
\bear
\Omega^1_{\proj 1 \times \proj 3} \cong \pi^\ast_1 \Omega^1_{\proj 1} \oplus \pi^\ast_3 \Omega^1_{\proj 3} = \Omega^1_{\proj 1} \boxplus \Omega^1_{\proj 3},
\eear
where now $\Omega^1_{\proj 1} \cong \mathcal{O}_{\proj 1} (-2).$ Note that, taking the dual of the Euler sequence~\eqref{euler}, one finds 
\bear
\xymatrix{
	0 \ar[r] &\Omega^1_{\proj n} \ar[r] & \mathcal{O}_{\proj n} (-1) \otimes V^\vee_{n+1} \ar[r] & \mathcal{O}_{\proj n} \ar[r] & 0,
}
\eear
which in turn describes the cotangent bundle on any projective space $\P^n.$

\subsubsection{Normal and conormal bundles} \label{ncnbundle}

\noindent Let us now consider $Y$ via its Segre embedding $\sigma: Y \hookrightarrow \proj 7.$ (We recall from \S\ref{prelim} that this embedding is canonically associated to the datum of the supertranslation algebra $\ft$.) Having introduced the tangent bundle $\mathcal{T}_Y$, one defines the normal bundle $\mathcal{N}_{Y/\proj 7}$ of $Y$ in $\mathbb{P}^7$ to be quotient bundle $\mathcal{T}_{\proj 7}|_Y / \mathcal{T}_Y$, where $\mathcal{T}_{\proj 7}|_Y \defeq \sigma^\ast \mathcal{T}_{\proj 7}.$ As such, the normal bundle sits in the exact sequence 
\bear \label{normalbundle}
\xymatrix{
	0 \ar[r] & \mathcal{T}_Y \ar[r]^{d\sigma \quad} & \mathcal{T}_{\proj 7}|_Y \ar[r] & \mathcal{N}_{Y / \proj 7} \ar[r] & 0,
}
\eear
of vector bundles on $Y$,
which will be referred to as the normal bundle exact sequence. Dualizing~\eqref{normalbundle}, 
one obtains the exact sequence defining the conormal bundle:
\bear \label{conormalbundle}
\xymatrix{
	0 \ar[r] & \mathcal{N}^\vee_{Y/\proj 7} \ar[r] & \Omega^1_{\proj 7}|_Y \ar[r]^{d\sigma^\vee} & \Omega^1_{Y} \ar[r] & 0.
}
\eear
The conormal bundle $\mathcal{N}^\vee_{Y/\proj 7} $ is thus the kernel of the morphism of vector bundles $d\sigma^\vee: \Omega^1_{\proj 7}|_Y \rightarrow \Omega^1_{Y}$.
Another characterization of the conormal bundle $\mathcal{N}^\vee_{Y / \proj 7}$ is possible using the sheaf of ideals $\mathcal{J}_Y $, which is defined as the kernel of the morphism of sheaves $\sigma^\sharp : \mathcal{O}_{\proj 7} \rightarrow \sigma_\ast \mathcal{O}_{Y}$. In fact, there is a natural isomorphism of vector bundles on $Y$ given by $\sigma^\ast (\mathcal{J}_Y / \mathcal{J}_Y^2) \cong \mathcal{N}^\vee_{Y/\proj 7}.$ In the following, since no confusion regarding the ambient space can arise, we will denote the normal and conormal bundles with respect to the Segre embedding by $\cN_Y$ and $\cN_Y^\vee$.

\section{A family of multiplets from line bundles}
\subsection{General procedure}
Let us now classify all multiplets associated to the infinite family of line bundles $\cO(n,m)$. We will denote the multiplets by 
\deq{
    \mu A^\bu(n,m) \defeq \mu A^\bu(\Gamma_*(\O(n,m))).
}
    As a first observation, we note that the construction exhibits the following symmetry under twists of line bundles\footnote{In general, this follows from the relation between Serre twists and internal grading, see \eqref{Gammast} and the comment right after.}:
\begin{equation} \label{eq: shift}
\begin{split}
\Gamma_*(\cO(n+k,m+k)) &= \bigoplus_{d \in \Z} H^0(\cO(n+k+d,m+k+d)) \\
&= \Gamma_*(\cO(n,m)) (k).
\end{split}
\end{equation}
This implies that the multiplets $\mu A^\bu(n,m)$ and $\mu A^\bu(n+k,m+k)$ agree up to a total degree shift. 
Since the weight grading of a graded equivariant $R/I$-module becomes the cohomological grading of the corresponding multiplet, we have that
\deq{
    \mu A^\bu(n+k,m+k) = \mu A^\bu\left( \vphantom{3^3} \Gamma_*(\cO(n,m)) (k) \right) 
    = \mu A^\bu\left( \vphantom{3^3} \Gamma_*(\cO(n,m)) \right)[k] = \mu A^\bullet(n,m)[k].
}
It is thus sufficient to consider the line bundles $\cO(n,0)$ and $\cO(0,m)$ for $n,m \geq 0$. (Equivalently, one could also consider the family $\cO(n,0)$ for $n \in \ZZ$.)

We will identify the field content of the multiplets using the technique sketched above in~\S\ref{hilbser}. We resum the equivariant Hilbert series, working in the ring of formal power series with coefficients in the representation ring of $\mathfrak{sl}(2) \times \mathfrak{sl}(4)$, and read off the equivariant structure of the minimal free resolution from its numerator.

Recall that $\Gamma_*(\cO(n,m))_d = \C[x_0,x_1]_{n+d} \otimes \C[y_0,\dots,y_3]_{m+d}$.
The monomials of degree $d$ are the $d$-th symmetric power of the defining representation of the corresponding group of linear transformations, so that we have 
\begin{equation}
\Gamma_*(\cO(n,m))_d = [n+d|m+d,0,0].
\end{equation}
in terms of Dynkin labels for $\mathfrak{sl}(2) \times \mathfrak{sl}(4)$.
Thus the equivariant Hilbert series takes the form of a formal Laurent series
\begin{equation} \label{eq: hilbert series}
    \Hilb(n,m) \defeq \Hilb(\Gamma_*(\mathcal{O}(n,m))) = t^{\mathrm{min}(n,m)} \sum_{d = 0}^{\infty} [n+d|m+d,0,0] \; t^d .
\end{equation}
Following~\S\ref{sec: prelim ps}, we rewrite the Hilbert series using the identity 
\begin{equation}
    \Hilb(n,m) = t^{\mathrm{min}(n,m)} \Hilb(R) \cdot \Hilb\left(\chi(W^\bu(n,m))\vphantom{3^3}\right),
\end{equation}
and solve for $\chi(W^\bu(n,m))$. The equations~\eqref{eq: recursion gen} become
\begin{equation} \label{eq: recursion}
\begin{split}
    \chi(W^\bu(n,m))_0 &= [n|m,0,0] \\
\chi(W^\bu(n,m))_1 &= [n+1|m+1,0,0] - [1|1,0,0] \otimes \chi(W^\bu(n,m))_0 \\
\vdots& \\
\chi(W^\bu(n,m))_k &= [n+k|m+k,0,0] - \sum_{d=1}^{k} \Sym^d([1|1,0,0]) \otimes \chi(W^\bu(n,m))_{k-d} \: .
\end{split}
\end{equation}
In what follows, we solve these equations case by case.
We will often tabulate our results by writing a square table representing~$H^\bu(K^\bu(\Gamma))$, or equivalently the bundle $\mu A^\bu(\Gamma)^\#$. In such tables, the (nonpositive) cohomological grading is on the horizontal axis and decreases to the right, while the weight grading on Koszul homology increases downward. As such, when reading the table as representing a multiplet, cohomological degree (ghost number) is on the \emph{vertical} axis, while parity is determined by the position on the horizontal axis modulo two. (For further notes on this convention, see~\cite[\S2.1]{EHSequiv}.)

\subsection{The bundles $\cO(n,0)$ for~$n\geq 0$} \label{sec:(n,0)} \label{On}

We begin with the case of the bundles $\cO(n,0)$ for nonnegative $n$. As we will see, these bundles include the vector multiplet, its antifield multiplet, and the hypermultiplet, as well as an infinite family of strict component-field multiplets associated to $\cO(n,0)$ with $n\geq 3$.

\subsubsection{Computation of the Betti numbers}
We specialize the Hilbert series~\eqref{eq: hilbert series} to the case at hand. At the level of the graded dimension,
\begin{equation}
\grdim(n,0) = \sum_{d = 0}^{\infty} (n+d+1) \frac{(d+3)(d+2)(d+1)}{6} \; t^d \:,
\end{equation}
which can be rewritten as a derivative of a geometric series
\begin{equation}
\grdim(n,0) = \frac{1}{6} \frac{\partial^3}{\partial t^3} t^{3-n} \frac{\partial}{\partial t}\sum_{d = 0}^{\infty} t^{d+n+1}
= \frac{1}{6} \frac{\partial^3}{\partial t^3} t^{3-n} \frac{\partial}{\partial t} \frac{t^{n+1}}{1-t}  .
\end{equation}
Performing the derivatives, the general result can be expressed in the following form.
\begin{equation}
\grdim(n,0) = \frac{(n+1) - 4nt + 6(n-1)t^2 - 4(n-2) t^3 + (n-3)t^4}{(1-t)^8} 
\end{equation}
The coefficients of the numerator now correspond to the Betti numbers of the associated multiplet.

Let us write out these Betti numbers concretely for all~$n$. These tables (here and in the following, may they contain Betti numbers or representations), follow the conventions as laid out in~\S\ref{p:gradings}, i.e. the vertical axis corresponds to cohomological degree and the horizontal axis to totalized degree. There are three special cases when $n \in \{0,1,2\}$. For $n=0$ one finds
\begin{equation}
\grdim\mu A^\bu(0,0)^\# = 	\begin{bmatrix}
1 &  &  &  \\
 & 6 & 8 & 3\\
\end{bmatrix} ,
\end{equation}
which corresponds to the vector multiplet. For $n=1$, we obtain
\begin{equation}
\grdim\mu A^\bu(1,0)^\# = 	\begin{bmatrix}
2 & 4 &  &  \\
 &  & 4 & 2\\
\end{bmatrix} ,
\end{equation}
which corresponds to the hypermultiplet. For $n=2$, the result reads
\begin{equation}
\grdim\mu A^\bu(2,0)^\# = 	\begin{bmatrix}
3 & 8 & 6 &  \\
 &  &  & 1\\
\end{bmatrix} ,
\end{equation}
which corresponds to the antifield multiplet of the vector multiplet. Finally, for $n \geq 3$, the resulting Betti numbers take the general form
\begin{equation}
\grdim\mu A^\bu(n,0)^\# = 	\begin{bmatrix}
n+1 & 4n & 6(n-1) & 4(n-2) & n-3
\end{bmatrix}  .
\end{equation}

\subsubsection{Equivariant decomposition}

The above recursive relations \eqref{eq: recursion} are easily solved, either by hand or with the help of a computer program such as LiE~\cite{LiE}. Let us again first consider the three special cases where $n \in \{0,1,2\}$. For $n=0$ we obtain
\begin{equation}
\begin{split}
W_0 &= \phantom{-}[0|0,0,0] \\
W_1 &= \phantom{-}0 \\
W_2 &= -[0|0,1,0] \\
W_3 &= \phantom{-}[1|0,0,1] \\
W_4 &= -[2|0,0,0] .
\end{split}
\end{equation}
Thus the resulting multiplet takes the form
\begin{equation}
	\mu A^\bu(0,0)^\# = \left[
	\begin{tikzcd}[row sep=0.3cm, column sep=0.2cm]
	\Omega^0 &  &  &  \\
	 & \Omega^1 &  \C^2 \otimes S_- & \Omega^0 \otimes \C^3
	\end{tikzcd} \right]
\end{equation}
where the three scalar fields live in the adjoint representation of the $R$-symmetry group. (Here and in the following tables showing the field content of multiplets $\C^n$ will always denote the unique irreducible $n$-dimensional representation of $\mathfrak{sl}(2)$.) This corresponds to the vector multiplet of six-dimensional $\cN=(1,0)$ supersymmetry. For $n=1$, we find
\begin{equation}
\begin{split}
W_0 &= \phantom{-}[1|0,0,0] \\
W_1 &= -[0|1,0,0] \\
W_2 &= \phantom{-}0 \\
W_3 &= \phantom{-}[0|0,0,1] \\
W_4 &= -[1|0,0,0] .
\end{split}
\end{equation}
We can thus identify $\mu A^\bu(1,0)$ as the hypermultiplet
\begin{equation} \label{hypm}
\mu A^\bu(1,0)^\# = \left[
\begin{tikzcd}[row sep=0.3cm, column sep=0.2cm]
\Omega^0 \otimes \C^2 & S_+ & & \\
& & S_- & \Omega^0 \otimes \C^2
\end{tikzcd} \right]
\end{equation}
For $n=2$
\begin{equation}
\begin{split}
W_0 &= \phantom{-}[2|0,0,0] \\ 
W_1 &= -[1|1,0,0] \\ 
W_2 &= \phantom{-}[0|0,1,0] \\
W_3 &= \phantom{-}0 \\
W_4 &= -[0|0,0,0] .
\end{split}
\end{equation}
The resulting multiplet $\mu A^\bu(2,0)$ is the antifield multiplet of the vector multiplet.
\begin{equation} \label{mu2}
\mu A^\bu(2,0)^\# = \left[
\begin{tikzcd}[row sep=0.3cm, column sep=0.2cm]
\Omega^0 \otimes \C^3 & S_+ \otimes \C^2 & \Omega^1 & & \\
& &  & \Omega^0
\end{tikzcd} \right]
\end{equation}
Finally for $n\geq 3$, the general form is
\begin{equation}
\begin{split}
W_0 &= \phantom{-}[n|0,0,0] \\ 
W_1 &= -[n-1|1,0,0] \\  
W_2 &= \phantom{-}[n-2|0,1,0] \\ 
W_3 &= -[n-3|0,0,1] \\ 
W_4 &= \phantom{-}[n-4|0,0,0] .
\end{split}
\end{equation}
Thus, $\mu A^\bu(n,0)$ for $n \geq 3$ are of the form
\begin{gather} \label{mu3}
\mu A^\bu(n,0)^\# = \\
\left[
    \begin{tikzcd}[row sep=0.3cm, column sep=0.2cm, ampersand replacement = \&]
\C^{n+1} \& \C^{n} \otimes S_+ \& \C^{n-1} \otimes \wedge^2 S_+ \& \C^{n-2} \otimes \wedge^3 S_+ \& \C^{n-3} \otimes \wedge^4 S_+
\end{tikzcd} 
\right] \notag
\end{gather}
This family of multiplets was described in the physics literature under the name $\cO(n)$-multiplets~\cite{KNTderivative, KuzenkoNovakTheisen, RulesProjective, LinchSugraProj, Galperin86}.

\subsubsection{Supersymmetry module structure and interpretation}
In this section, we will use the results obtained in~\S\ref{On} to describe the supersymmetry module structure of the multiplet $\mu A^\bullet (n,0) $ associated with the line bundles $\mathcal{O} (n,0)$. 

In order do to this, we start by giving an explicit presentation of the modules $\Gamma_*(\cO(n,0))$ as a cokernel of a map between free $R$-modules, and go on to describe their minimal free resolutions in $R$-modules. The cases $n=0$ and $n=1$ were already discussed in~\cite[\S 6.4]{perspectives}. Recalling that we have denoted with $\mathbb{C}^n$ the unique irreducible $n$-dimensional representation of $\mathfrak{sl}(2)$, we identify  $\mathbb{C}^{n+1} \cong \mbox{Sym}^n U $, where $U = \C^2$ is the fundamental representation of $\mathfrak{sl}(2)$.  By looking at~\eqref{hypm}, one can write down the first map (differential) appearing in the resolution of $\Gamma_\ast (\mathcal{O} (1,0))$ in $R$-modules as
\begin{equation} \xymatrix@R=1.5pt{
\varphi_1 : S_+ \otimes R \ar[r] &  U \otimes R \\
F_\alpha s^\alpha  \ar@{|->}[r] & \varphi_1 (F_\alpha s^\alpha ) \defeq \lambda^\alpha_{i} F_\alpha e^i, 
}
\end{equation} 
where we have chosen bases $\{ s^\alpha \}_{\alpha = 1, \ldots, 4}$ and $\{e^i\}_{i = 1,2}$ for $S_+$ and $U$ respectively. Up to a constant, this is the unique equivariant map which is linear in $\lambda$, see~\cite[\S 6.4]{perspectives}. With a standard choice of bases, this linear map is represented by the matrix 
\begin{equation}
\varphi_1 = \left ( \begin{array}{cccc} 
\lambda^1_1 & \lambda^2_1 & \lambda^3_1 & \lambda^4_1  \\
\lambda^1_2 & \lambda^2_2 & \lambda^3_2 & \lambda^4_2
\end{array}
\right ),
\end{equation} 
and one sees that its image contains $I \otimes U$, where $I$ is the defining ideal of the nilpotence variety. As a consequence, $\Gamma_\ast (\mathcal{O} (1,0)) \cong \coker ({\varphi_1})$, as (equivariant) $R/I$-modules.
This can be generalized as follows for $\Gamma_\ast (\mathcal{O} (n, 0))$. 
\begin{prop}
    For $n\geq 1$, there is an isomorphism $\phi_n$ of $(\mathfrak{sl}(2) \times \mathfrak{sl}(4))$-equivariant $R/I$-modules $\Gamma_*(\cO(n,0)) \cong \coker(\varphi_n)$, linear in the generators, defined by 
	\begin{equation} \label{varphin} 
            \begin{aligned}
	\varphi_n : \C^n \otimes S_+ \otimes R &\rightarrow  \C^{n+1} \otimes R \\
	F_{(i_1 \ldots i_{n-1})\:\alpha}   &\mapsto 
        \lambda^\alpha_{(i_n} F^{\phantom{\alpha}}_{i_1 \dots i_{n-1}) \: \alpha}
    \end{aligned}
	\end{equation}
        in terms of bases $\{ s^\alpha \}$ for $S_+$ and $\{ e^{i_1} \odot \ldots \odot e^{i_{n-1}} \}$ for $\C^n \cong \Sym^{n-1} U$.
\end{prop}
\begin{proof}
   	The result follows from our descriptions of the component field multiplets in the previous section and their relation to the minimal free resolution of the modules $\Gamma_*(\cO(n,0))$. In particular, from \eqref{hypm}, \eqref{mu2} and \eqref{mu3}, it follows that $\Gamma_*(\cO(n,0))$ can be written as the cokernel of a map with domain and codomain as specified in the proposition. This map is just the first piece of the resolution differential. In addition, degree reasons imply that this map is linear in the variables $\lambda^\alpha_i$ and that $I \otimes \C^{n+1}$ is contained in the image of $\varphi_n$, so that in particular the map descends to a map of $R/I$-modules. Moreover, a quick representation theoretic check shows that the map is equivariant under $\mathfrak{sl}(2) \times \mathfrak{sl}(4)$, and it is in fact the unique map with these properties, up to a non-zero constant prefactor.
\end{proof}
Resolving $\Gamma_*(\cO(n,0)) = \mathrm{coker}(\varphi_n)$ one recovers the field content of the multiplets described above. In addition, the resolution differential encodes the part of the $\fp$-module structure acting by differential operators of degree zero (this is the $\mathrm{Gr}(\fp)$-module structure induced from the $\fp$-module structure)~\cite[\S\S3.6--7 and Theorem 4.8, respectively]{EHSequiv, perspectives}. 

Let us describe the minimal free resolution and the module structure for the cases $n\geq 3$. This will provide an intuitive interpretation of $\mu A^\bu(n,0)$: it is a multiplet whose observables are generated by the degree-$n$ monomials in the observables of the $\cO(1,0)$-multiplet, i.e. the hypermultiplet. One can thus imagine that the fields of the hypermultiplet map to the fields of the $\cO(n,0)$ multiplet via a (ramified) $n$-fold covering, dual to the inclusion map on observables. 

The minimal free resolution of $\Gamma_\ast (\cO (n,0))$ for $n\geq 3$ is given by (see~\eqref{mu3}) 
\begin{equation}\xymatrix{
        L^\bullet \defeq  \bigg ( \C^{n+1} & \ar[l]_{\quad \; \; (d_L)_1} \C^n \otimes S_+  & \ar[l]_{(d_L)_2} \C^{n-1} \otimes \wedge^2 S_+  & \ar[l]_{\; \; (d_L)_3} \C^{n-2} \otimes \wedge^3 S_+  & \ar[l]_{(d_L)_4 \quad } \C^{n-3} \otimes \wedge^4 S_+  \bigg ) \otimes R,
}
\label{eq:O(n)}
\end{equation}
where the resolution differentials (notice that $(d_L)_1 = \varphi_n$) 
\begin{equation}
(d_L)_i : \C^{n-i} \otimes \wedge^i S_+ \otimes R \longrightarrow \C^{n-i+1} \otimes \wedge^{i-1} S_+ \otimes R \qquad i=1 \dots 4 
\end{equation}
are described by contracting along $S_+$ and symmetrizing along the $\mathfrak{sl}(2)$-representations, for example
\begin{equation}
[(d_L)_2 F]_{i_1 \dots i_{n-1} \alpha} = \lambda^\beta_{(i_{n-1}} F^{\phantom{\beta}}_{i_1 \dots i_{n-2}) \: \alpha \beta} \: ,
\end{equation} 
generalizing~\eqref{varphin}. This translates into supersymmetry transformation rules of the form
\begin{equation}
\delta F_{i_1 \dots i_n} = \epsilon^\alpha_{(i_n} F^{\phantom{\alpha}}_{i_1 \dots i_{n-1) }\alpha} \: .
\end{equation}
Recall that we identified $\mu A^\bu(1,0)$ as the hypermultiplet. Let us denote the linear observables in physical fields of the hypermultiplet by $\phi_i$ and $\psi_\alpha$. This suggests to identify the linear observables of the $\cO(n,0)$-multiplet as polynomials of degree $n$ in the linear observables of the hypermultiplet, as follows:
\begin{equation}
\begin{split}
F_{i_1 \dots i_n} &= \phi_{i_1} \dots \phi_{i_n} \\
F_{i_1 \dots i_{n-1} \: \alpha} &= \phi_{i_1} \dots \phi_{i_{n-1}} \psi_\alpha \qquad \\
\vdots \\
F_{i_1 \dots i_{n-4} \: \alpha \beta \gamma \delta} &= \phi_{i_1} \dots \phi_{i_{n-4}} \psi_\alpha \psi_\beta \psi_\gamma \psi_\delta .
\end{split}
\end{equation}
Further, recall that for the hypermultiplet the module structure of the supersymmetry algebra contains terms of the form~\cite[\S 6.4 and (6.11)]{perspectives}
\begin{equation}
\delta \phi_i = \epsilon_i^\alpha \psi_\alpha \: .
\end{equation}
By the Leibniz rule, this precisely induces the supersymmetry transformations of the $\cO(n,0)$-multiplet we recorded above. Thus, we can view, for $n\geq3$, $\mu A^\bu(n,0)$ as consisting of polynomials of degree $n$ in the linear observables of $\mu A^\bu(1,0)$. Intuitively, this can be viewed as a remnant of the statement $\cO(n,0) = \cO(1,0)^{\otimes n}$ after applying the pure spinor superfield formalism. We remark that a special case of this is already visible in the action for supersymmetric Yang--Mills theory coupled to hypermultiplets studied in~\cite{Ced-6d}. There, an action is written that reproduces the minimal coupling of the gauge sector to matter; the relevant term is cubic, containing two hypermultiplets and one gauge field. From our perspective, this makes use of the identification of the $\cO(2,0)$ multiplet both as the dual to the vector multiplet and as governing quadratic functionals on the hypermultiplet. 

\numpar[p:twist][Holomorphic twists of $\O(n)$-multiplets]
It is straightforward to compute the holomorphic twist of these multiplets uniformly for $n\geq 3$.
Following Remark~\ref{rmk:twists},
we expect to find that the twist is of rank one over Dolbeault forms on~$\C^3$. It should thus resolve the holomorphic sections of a holomorphic line bundle on spacetime. Our result shows that this holomorphic line bundle is the $n$-th power of the square root of the canonical bundle arising from the spin structure.
\begin{thm}
    The holomorphic twist of the multiplet $\mu A^\bu(n,0)$ for $n\geq 3$ is given by the chain complex of vector bundles on~$\C^3$ whose sections are 
    \deq{
        \left( \Omega^{0,\bu}(\C^3, K^{n/2}), \, \dbar \right).
    }
\end{thm}
We remark that the statement of the theorem in fact holds for $n\geq 0$. The results for $n = 0, 1, 2$---the vector multiplet, the hypermultiplet, and the antifield multiplet of the vector multiplet---are well-known. We view this result as a demonstration that $A^\bu(2,0)$ can be thought of as playing the role of the canonical bundle of $\N=(1,0)$ superspace, viewed as the affine dg scheme $\Spec A^\bu$.
\begin{proof}
    Let $U = \C^2$ denote the R-symmetry space, which is the fundamental representation of $\lie{sp}(1) = \lie{sl}(2)$. We recall from~\eqref{mu3} that the multiplet consists of the bundles $\Sym^{n-k}(U) \otimes \wedge^k(S_+)$, in cohomological degree zero and totalized degree $-k$, for $0 \leq k \leq 4$.

    Choosing a holomorphic supercharge $Q$ fixes a complex structure on~$\R^6$ (specified by a maximal isotropic subspace $T$) and a polarization of~$U$ (equivalently, a choice of Cartan subalgebra of~$\lie{sl}(U)$).
    With respect to the flat K\"ahler metric on~$\C^3$, we can identify  
\deq{
    S_- \cong \left( \underline{\C} \oplus \wedge^2 T \right) \otimes K^{1/2}, \quad
    S_+ \cong \left( T \oplus \wedge^3 T \right) \otimes K^{1/2}
}
as bundles on~$\C^3$. It will be convenient to 
write $S_+$ as $K^{-1/2} \oplus \left( T \otimes K^{1/2} \right)$, so that 
\deq{
    \wedge^k(S_+) \cong \left( K^{k/2} \otimes \wedge^k {T} \right) \oplus \left( K^{k/2 - 1} \otimes \wedge^{k-1} {T} \right).
}

The supercharge $Q$ carries cohomological degree zero, 
carries weight $+1$ for the Cartan of $\lie{sl}(U)$, and transforms as a section of $K^{1/2}$. To twist the theory, we choose twisting data that regrade the theory to place $Q$ in cohomological degree $+1$, while twisting the representation of the holonomy group to make $Q$ a scalar. This amounts to the replacement 
\deq{
    U \leadsto K^{-1/2}[-1] \oplus K^{1/2}[1].
}
From this, we identify
\deq{
    \Sym^m(U) \leadsto K^{-m/2} \otimes \left( \bigoplus_{\ell = 0}^m K^{\ell}[2\ell - m]\right) .
}

We collect the summands in the multiplet proportional to $\wedge^k {T}$ for $k = 0, \ldots, 3$. With $k=0$, we have contributions 
\deq{
    K^{-n/2} \otimes \left( \bigoplus_{\ell = 0}^{n} K^{\ell}[2\ell-n]\right)  
    \oplus K^{-n/2} \otimes \left( \bigoplus_{\ell = 0}^{n-1} K^{\ell}[2\ell - n + 1]\right) .
}
The non-derivative supersymmetry transformations indicated above map the first set of summands onto the second set, leaving only the summand $K^{n/2}[n]$ from the $\ell = n$ term in cohomology.

The terms proportional to $\wedge^1 {T}$ take the form 
\deq{
    K^{-n/2+1} \otimes {T} \otimes  \left( \bigoplus_{\ell = 0}^{n-1} K^{\ell}[2\ell-n+1]\right)
    \oplus K^{-n/2+1} \otimes {T} \otimes \left( \bigoplus_{\ell = 0}^{n-2} K^{\ell}[2\ell-n+2]\right) .
}
Again, the non-derivative supersymmetry transformations act acyclically wherever possible, leaving only $K^{n/2}[n-1] \otimes {T}$ from the $\ell = n$ summand in cohomology. By identical computations, we see that the underlying bundle of the twisted multiplet reduces to 
\deq{
    \bigoplus_{k=0}^3 K^{n/2} \otimes \wedge^k {T}[n-k],
}
after passing to the cohomology of all acyclic differentials.

Recalling that $T \cong \bar{T}^\vee$, we identify the sections of this graded bundle with $\Omega^{0,\bu}(K^{n/2})[n]$. The remaining terms (differential operators of order one) in the holomorphic supercharge generate precisely the $\dbar$ operator, so that the twisted multiplet can be identified as
\deq{
    \mu A^\bu(n,0)^Q \simeq \Omega^{0,\bu}(\C^3, K^{n/2})[n].
}
\end{proof}

\subsection{The bundles $\cO(0,m)$ for~$m\geq 0$} \label{sec: O(0,m)}
\subsubsection{Computation of the Betti numbers}
The Hilbert series specializes to
\begin{equation}
\Hilb(0,m) = \sum_{d = 0}^{\infty} (d+1) \frac{(m+d+3)(m+d+2)(m+d+1)}{6} \; t^d \: ,
\end{equation}
which can be rewritten as 
\begin{equation}
\Hilb(0,m) = \frac{1}{6} \frac{\partial}{\partial t} t^{1-m} \frac{\partial^3}{\partial t^3} \frac{t^{m+3}}{1-t} \: .
\end{equation}
Again, we can bring the Hilbert series into a form such that we can read off the Betti numbers of the associated multiplet.
\begin{equation}
\begin{split}
\frac{1}{(1-t)^8} &\Bigg[ (\frac{m^3}{6} + m^2 +\frac{11}{6}m +1) - (m^3 + 5m^2 +6m)t + (\frac{5}{2}m^3 + 10m^2 + \frac{11}{2}m - 6)t^2\\ &- (\frac{10}{3}m^3 + 10m^2- \frac{4}{3}m -8)t^3
+ (\frac{5}{2}m^3 + 5m^2 -\frac{9}{2}m - 3)t^4 - (m^3+m^2 -2m)t^5 + (\frac{m^3}{6} - \frac{m}{6})t^6 \Bigg]
\end{split}
\end{equation}
It is immediate to see that for $m=0$ we recover the result from above. Let us in addition give the Betti tables for some small values of $m$. For $m=1$, we find
\begin{equation}
\grdim\mu A^\bu(0,1) = 	\begin{bmatrix}
4 & 12 & 12 & 4 \\
\end{bmatrix} \: .
\end{equation}
For $m=2$ one obtains
\begin{equation}
\grdim\mu A^\bu(0,2) = 	\begin{bmatrix}
10 & 40 & 65 & 56 & 28 & 8 & 1 \\
\end{bmatrix} \: .
\end{equation}

\subsubsection{Equivariant decomposition}
Solving the equations~\eqref{eq: recursion} one finds the following representations appearing in $\mu A^\bu(0,m)$.
\begin{equation}
\begin{split}
W_0 &= \phantom{-}[0|m,0,0] \\
W_1 &= -[1|m-1,1,0] \\
W_2 &= \phantom{-}[0|m-2,2,0] + [2|m-1,0,1] \\
W_3 &= -[1|m-2,1,1] -[3|m-1,0,0] \\
W_4 &= \phantom{-}[0|m-2,0,2] + [2|m-2,1,0] \\
W_5 &= -[1|m-2,0,1] \\
W_6 &= \phantom{-}[0|m-2,0,0]
\end{split}
\end{equation}

\subsubsection{Presentation and equivariant resolution}
We can describe the module $\Gamma_*(\cO(0,1))$ explicitly as the cokernel of a map of free $R$-modules
\begin{equation}
	\psi_1 : (\wedge^2 S_+ \otimes \C^2) \otimes R \longrightarrow S_+ \otimes R \: .
\end{equation}
For degree reasons, the map should be linear in $\lambda$. It is easy to check that there is, up to non-zero constant prefactors, a unique such map explicitly given by
\begin{equation}
	G \mapsto \lambda^\alpha_i G^i_{[\alpha \beta]} s^\beta \: .
\end{equation}
Here $s^\beta$ denotes a basis of $S_+$. The modules $\Gamma_*(\cO(0,m))$ are obtained by taking symmetric products. It can be checked explicitly (for example using a computer program such as \textit{Macaulay2}~\cite{M2}) that the minimal free resolutions of these modules reproduce the multiplets described above. 

\subsection{A classification result}
The results above describe all multiplets for six-dimensional $\cN=(1,0)$ supersymmetry which can be obtained from line bundles on the nilpotence variety $\P^1 \times \P^3$. Based on the equivalence of categories between multiplets and $C^\bullet(\ft)$-modules provided by~\cite[Theorem 4.3]{EHSequiv} this can be viewed as a classification result. More precisely, recalling the definition of a multiplet $(E, D, \rho)$ as \cite[Definition 2.11]{EHSequiv} and denoting with $\mathcal{E}$ the global sections of the vector bundle $E$ on (flat) spacetime characterizing the multiplet, one has the following.  
\begin{thm}
	The above multiplets $\mu A^\bu(n,m)$ classify, up to quasi-isomorphism, all multiplets $(E, D, \rho)$ for six-dimensional $\cN=(1,0)$ supersymmetry such that $H^\bullet(\ft,\cE)$ is the graded global section module of a single line bundle on the projective nilpotence variety.
\end{thm}
\begin{proof} The proof is based on~\cite[Theorem 4.3]{EHSequiv} and subsequent corollaries. 
    Recall that any line bundle $\mathcal{L}$ on $Y$ is of the form $ \mathcal{O} (n,m)$ for some integers $n,m$ as $\mbox{Pic}\, (Y) = \mathbb{Z}\times \mathbb{Z}$, and let $\Gamma_\ast (\mathcal{O} (n,m))$ be its associated $R/I$-module. Given an $R/I$-module $\Gamma$, the derived $\ft$-invariants of the associated multiplet are concentrated in a single degree. By construction,
\begin{equation}
	H^\bullet(\ft, \mu A^\bu(\Gamma)) = \Gamma \: .
\end{equation}
In particular, taking $\Gamma = \Gamma_\ast (\mathcal{O}_Y (n,m))$ and denoting as above $\mu A^\bu(n,m)$ the associated multiplet, one has that $H^\bullet(\ft, \mu A^\bu(n,m)) = \Gamma_\ast (\mathcal{O}_{Y} (n,m))$.
Conversely, given a multiplet $(E,D,\rho)$ such that its derived $\ft$-invariants are concentrated in a single degree, one can identify
\begin{equation}
	\mu A^\bu(C^\bullet(\ft, \cE)) \simeq (E,D,\rho) \: .
\end{equation}
It is enough to take $(E,D,\rho)$ of the form $\mu A^\bullet(n,m)$ to conclude the verification, since their $\mathfrak{t}$-invariants are $R/I$-modules.
\end{proof}
As remarked above (Remark~\ref{rmk:twists}), the interpretation of the input module as the Chevalley--Eilenberg cohomology with coefficients in the multiplet provides an interesting conceptual link to the twists of the multiplet involved. Twisting by a supercharge $Q$ takes invariants of the multiplet with respect to the abelian subalgebra spanned by that supercharge. The cohomology groups $H^\bullet(\ft,\cE)$ define a sheaf on the nilpotence variety which, by the result of~\cite{spinortwist}, encodes all the information on the twists of the original multiplet. In fact, one expects that the twist by a square-zero supercharge $Q \in Y$ is determined by the stalk of that sheaf at $Q$.

In our example, we see that---as the derived invariants of all the multiplets above are line bundles---the stalk at any point is isomorphic to $\cO_{Y,x}$. Our nilpotence variety $Y=\P^1 \times \P^3$ only has one stratum corresponding to the holomorphic twist. Correspondingly, as we have seen above, the holomorphic twists of the above multiplets always have rank one over Dolbeault forms on~$\C^3$.

This intuition makes many aspects of the physical behavior of the multiplets and their twists manifest. For example, we can take any of the above multiplets and dimensionally reduce to a four-dimensional $\cN=2$ multiplet. In this case, the nilpotence variety is reducible and has three different components, one of which is the image of the six-dimensional $\cN=(1,0)$ nilpotence variety under the dimensional reduction map~\cite{NV}. The other two components correspond to the Donaldson--Witten twist, which does not descend from a square-zero supercharge in six dimensions. The Chevalley--Eilenberg cohomology with coefficients in the dimensionally reduced multiplets is obtained by pushing forward along the inclusion $Y_{6D} \hookrightarrow Y_{4D}$. Clearly, any supercharge corresponding to a Donaldson--Witten twist is outside of the support of the resulting sheaf, so that the respective stalks are trivial. Following Remark~\ref{rmk:twists}, one thus expects that the Donaldson--Witten twists of all multiplets arising by dimensional reduction are perturbatively trivial.  We hope to give a more complete account of extensions of the methods developed in~\cite{spinortwist} to general multiplets in future work.

\section{Antifield multiplets and duality}

\subsection{General observations}
Given any multiplet $\mu A^\bu(\Gamma)$, one may form the dual (or antifield) multiplet $\mu A^\bu(\Gamma)^\vee$ by dualizing the underlying vector bundle, the differential and the supersymmetry module structure (see~\cite{EHSequiv} for more information, see in particular Definition 2.18). Via the pure spinor superfield formalism, the operation of taking the antifield multiplet corresponds, in good cases, to taking the dualizing module of the input module $\Gamma$. This was already recognized in~\cite[\S\S4.1--3]{perspectives}. Here, we explore this direction further and link it to statements in terms of sheaves on the nilpotence variety.

\subsubsection{Duality theory for modules and multiplets}
Let us start by reviewing some facts from commutative algebra (see for example~\cite{Eisenbud}) and relate these to the pure spinor superfield formalism. 
\begin{dfn}
	Let $S$ be a commutative ring with unity and let $M$ be an $S$-module.
	We define the dimension of $M$ to be 
	\begin{equation}
	\mathrm{dim}_S (M ) \defeq \mathrm{dim}(S / \mathrm{Ann}_S (M)),
	\end{equation} 
	where $\mathrm{Ann}_S (M) \defeq \{ s \in S : s m = 0 \; \forall m \in M \}$. In other words, the dimension of $M$ is the Krull dimension of the quotient ring $S / \mathrm{Ann}_S (M)$.
	\end{dfn}			
The definition of the dimension of a module as given above might appear counterintuitive. On the other hand, recalling that the \emph{support} of $M$ is defined as $\mathrm{supp} (M) \defeq \{ \mathfrak{p} \in \mathrm{Spec} (S) : \; M_\mathfrak{p} \neq 0 \} \subseteq \mathrm{Spec} (S)$, where $M_{\mathfrak{p}}$ is the localization of $M$ at the prime ideal $\mathfrak{p}$, it is not hard to see that $\mathrm{supp} (M) = \{ \mathfrak{p} \in \mathrm{Spec} (S) :\;  \mathfrak{p} \supseteq \mathrm{Ann}_S (M) \} = V (\mathrm{Ann}_S (M))$\footnote{Here we assume $M$ to be finitely generated.}. Hence, it follows that $\mathrm{dim}_S (M) = \mathrm{dim}_S (S / \mathrm{Ann}_S (M)) = \mathrm{dim}_S (\mathrm{supp} (M)) $. 
Similarly, if $\mathcal{F}$ is a sheaf on a scheme $X$, then one defines the support of $\mathcal{F}$ as $\mathrm{supp} (\mathcal{F}) \defeq \{ x \in X : \mathcal{F}_x \neq 0 \}$, where $\mathcal{F}_x$ is the stalk of $\mathcal{F}$ at $x \in X$. If $M = \Gamma (U, \mathcal{F} )$ for an open set $U \subset X$ with $U = \mathrm{Spec} (S)$ and if $\mathfrak{p} \subset \mathrm{Spec} (S)$ is the prime ideal corresponding to a point $x \in U$, then one has that $\mathfrak{p} \in \mathrm{supp} (M) $ if and only if $x \in \mathrm{supp} (\mathcal{F}).$ The relation between the dimension and the support of a module---or the associated sheaf---conveys the idea that the support measures (in some sense) the size of $M$, as outside of it the fibers of $M$ are zero.
\begin{dfn} 	
	Let $S$ be a commutative ring with unity, let $M$ be an $S$-module, and let $I \subset S$ be an ideal in $S$ such that $I \cdot M \neq M$. 
	We define the $I$-depth of $M$ to be 
	\begin{equation}
	\mathrm{depth}_I (M) \defeq \mathrm{min} \{ i \in \mathbb{N} : \Ext^i_S (S/I, M) \neq 0\}.
	\end{equation}
\end{dfn}
Note that the definition of depth refers to a choice of ideal $I$ in~$S$.
In our case, we are interested in modules over the polynomial ring $R$. Maximal ideals then correspond to points $x \in \Spec R$; there exists a unique maximal equivariant ideal $\mathfrak{m}$ in~$R$, corresponding to the skyscraper sheaf at the origin, and consisting of all polynomials with zero constant term. In this case, given an $R$-module $M$ we always define its depth with respect to the unique maximal (equivariant) ideal $\mathfrak{m} \subset R$, \emph{i.e.} we define $\mathrm{depth}_R (M) \defeq \mathrm{depth}_{\mathfrak{m}} (M).$\footnote{This is the usual definition of depth of a module in the case of \emph{local} rings $(R, \mathfrak{m}).$} In this setting, the central object of our interest is the Cohen--Macaulay property for modules $R$.
\begin{dfn}	
	 We say that the $R$-module $M$ is Cohen--Macaulay if $\mathrm{depth}_R(M) = \mathrm{dim}_R(M)$.
\end{dfn}
In particular, let us consider polynomial rings of Krull dimension $\mathrm{dim} (R) = n = \dim \mathfrak{t}_1$, where $\mathfrak{t}_1$ is the odd part of the supertranslation algebra---explicitly, this is the case of the polynomial ring $\mathbb{C}[\lambda_1, \ldots, \lambda_n]$ that we will use in what follows. In this context, two equivalent characterizations of the Cohen--Macaulay property will be useful. The first one is in terms of the length of a minimal free resolution.\footnote{The length of a free resolution $L^\bu = (L^0 \leftarrow L^{-1} \leftarrow \dots \leftarrow L^{-k} \leftarrow 0)$ is $k$.} The Auslander--Buchsbaum formula implies the following.
\begin{prop}
	An $R$-module $\Gamma$ is Cohen--Macaulay if the length of its minimal free resolution equals its codimension, i.e.
	\begin{equation}
	l_R(\Gamma) = n - \mathrm{dim}_R(\Gamma) = \mathrm{codim}_R(\Gamma) \: .
	\end{equation}
\end{prop}
Equivalently we can characterize Cohen--Macaulay modules via their $\Ext$-groups.
\begin{prop}
	An $R$-module $\Gamma$ of dimension $q$ is Cohen--Macaulay if and only if $\Ext^k_R(\Gamma, R) = 0$ for all $k \neq \mathrm{codim}_R(\Gamma) = n-q \defeqinv r$.
\end{prop}
In the context of the pure spinor superfield formalism we consider $R/I$-modules $\Gamma$ as input data, which can be canonically looked at as $R$-modules via the quotient map $R \rightarrow R/I$. Given a Cohen--Macaulay module $\Gamma$, the multiplet $\mu A^\bu(\Gamma)$ is described by the minimal free resolution of $\Gamma$ in $R$-modules. The antifield multiplet $\mu A^\bu(\Gamma)$ is described by the dual of that minimal free resolution, which is, by definition, a minimal free resolution of the dualizing module $\Ext^r_R(\Gamma,R)$. Therefore we can identify for Cohen--Macaulay modules $\Gamma$, 
\begin{equation} \label{extr}
\mu A^\bu(\Gamma)^\vee = \mu A^\bu(\Ext^{r}_R(\Gamma,R)) \: .
\end{equation}
If $\Gamma$ is not Cohen--Macaulay, this is no longer true. Then the dual of the minimal free resolution of $\Gamma$ is no longer a resolution of a single module, but in fact a model for the dualizing complex of $\Gamma$. Its cohomology is the $\Ext$-algebra $\Ext^\bullet_R(\Gamma , R)$.

\subsubsection{The relation to sheaves}
In this work we focus on modules which arise from sheaves on the nilpotence variety $Y$ via $\Gamma_*$. In this setting, we can link the above statements on duality to more geometric notions for sheaves on projective schemes.

Therefore, let us consider a Cohen--Macaulay projective scheme $\iota: X \hookrightarrow \P^n$ of codimension $r$. In this setting the dualizing sheaf of $X$ is a vector bundle, denoted by $\omega_X^\circ$. Explicitly, it can be defined in terms of the ambient projective space as
\begin{equation}
	\omega^\circ_X = \SExt^r_{\cO_{\P^n}}(\iota_* \cO_X , \omega_{\P^n} ) \: .
\end{equation} 
Let us further assume that $\Gamma_*(\cO_X)$ is Cohen--Macaulay as an $R=\Gamma_*(\cO_{\P^n})$-module\footnote{Note also that by restriction of scalars $\Gamma_\ast : \Coh_{\cO_X} \rightarrow \Mod_{\Gamma_*(\cO_X)}$ can be seen as a functor mapping to $\Mod_R$.}. Then the following holds.
\begin{prop} \label{duality} Let $\SHom_{\mathcal{O}_{\scalemath{0.4}{\P}^{\scalemath{0.4}{n}}}} ( \iota_\ast- ,\, \omega^\circ_X) : \Coh_{\cO_X} \rightarrow \Coh_{\cO_{\P^n}}$ and $\Ext^{r}_R (- ,\, R) : \Mod_R \rightarrow \Mod_R$. Then the following diagram commutes 
\begin{equation}
\xymatrix{ \Coh_{\cO_X} \ar[dd]_{\SHom_{\mathcal{O}_{\scalemath{0.4}{\P}^{\scalemath{0.4}{n}}}} ( \iota_\ast- ,\, \omega^\circ_X)(n+1)} \ar[rr]^{\Gamma_\ast} & & \Mod_R \ar[dd]^{\Ext^{r}_R (- ,\, R)}  \\ \\
\Coh_{\cO_{\P^n}}  \ar[rr]^{\quad \Gamma_\ast}  & & \Mod_R.
}
\end{equation} 
In particular, let $\mathcal{F} \in \Coh(X)$ be a coherent sheaf on $X$ and $\Gamma_\ast (\mathcal{F})$ be its associated $R = \Gamma_*(\cO_{\P^n})$-module. The following is a natural isomorphism
\begin{equation}
\Ext^{r}_R (\Gamma_\ast (\mathcal{F}), R) \cong \Gamma_\ast \SHom_{\mathcal{O}_{\P^{\scalemath{0.4}{n}}}} (\iota_\ast \mathcal{F}, \omega^\circ_X)(n+1).
\end{equation}
\end{prop} 
\begin{proof} Since shifting does not change the cohomology class, one has
\begin{align}
\Ext^r_R (\Gamma_\ast (\mathcal{F}), R) &\cong \Ext^r_R (\Gamma_\ast (\mathcal{F}) (-n-1) \: , \:  R(-n-1)) \nonumber \\
& \cong  \Ext^r_R (\Gamma_\ast (\mathcal{F}) (-n-1) \otimes_R \Gamma_*(\cO_X) \: , \:  R(-n-1)) \: ,
\end{align}
where the second isomorphism follows from the fact that the sheaf $\mathcal{F}$ is supported on $X$ and $\widetilde{\Gamma_\ast (\mathcal{O}_X)}\cong \mathcal{O}_X.$  
By derived hom-tensor adjunction~\cite{Huybrechts} one has
\begin{equation} \label{eq: tensor hom}
\begin{split}
	\Ext^r_R \left (\Gamma_\ast (\mathcal{F}) , R \right ) &\cong \Hom_R \left (\Gamma_\ast (\mathcal{F}) (-n-1) \: , \: \Ext^r_R (\Gamma_*(\cO_X), R(-n-1)) \right ) \\
	&=  \Hom_R \left (\Gamma_\ast (\mathcal{F}) (-n-1) \: , \: \Ext^r_R (\Gamma_*(\cO_X), \Gamma_*(\omega_{\P^n})) \right ) \: ,
\end{split}
\end{equation}
where we used that $\Gamma_*(\omega_{\P^n}) = R(-n-1)$ in the second step. Notice that, by assumption $\Gamma_\ast (\mathcal{O}_X)$ is Cohen--Macaulay as an $R$-module, hence the only non-zero Ext-module in the derived hom-tensor adjunction is indeed the dualizing module $\Ext^r_R (\Gamma_*(\cO_X), R(-n-1))$.
Further, note that $\Hom_R$ in \eqref{eq: tensor hom} denotes graded morphisms of all degrees. For morphisms of degree zero, we have the adjunction~\cite{FOAG, Hartshorne}
\begin{equation}
	\Hom_R^{\deg =0}(M, \Gamma_*(\cH)) = \Hom_{\P^n}(\tilde{M} , \cH)
\end{equation}
between the functors $\Gamma_\ast : \QCoh_{\mathcal{O}_X} \rightarrow \Mod_R $ and $ \widetilde{(\cdot)} : \Mod_R \rightarrow \QCoh_{\mathcal{O}_X}$.
Upon using $\widetilde{\Gamma_{\ast}(\cG)} = \cG$, this implies that
\begin{equation}
	\Hom_R^{\deg =0}(\Gamma_*(\cG), \Gamma_*(\cH)) = \Hom_{\P^n}(\cG , \cH) \: .
\end{equation}
Shifting and summing on both sides, one reconstructs the graded morphisms:
\begin{equation}
\begin{split}
	\Hom_R(\Gamma_*(\cG), \Gamma_*(\cH)) &= \bigoplus_{k \in \ZZ} \Hom_{\P^n}(\cG , \cH (k))
	 =\bigoplus_{k \in \ZZ} \Gamma \circ \left( \SHom_{\P^n}(\cG , \cH (k)) \right) \\
	& =\bigoplus_{k \in \ZZ} \Gamma \circ \left( \SHom_{\P^n}(\cG , \cH) \otimes \mathcal{O}_{\P^n}(k) \right) \\
	& =\Gamma_*(\SHom_{\P^n}(\cG , \cH))  ,
\end{split}
\end{equation}
where we have used that $\Gamma \circ \SHom_{\P^n} = \Hom_{\P^n}, $ where $\Gamma$ is the global section functor. Deriving the above functors, one gets a local-to-global spectral sequence, which in the case of interest yields the isomorphism 
\begin{equation}
	\Ext_R^r(\Gamma_*(\cO_X), \Gamma_*(\omega_{\P^n})) \cong \Gamma_*(\SExt^r_{\P^n}(i_*\cO_X, \omega_{\P^n})) \: .
\end{equation}
Plugging this into~\eqref{eq: tensor hom}, we finally obtain
\begin{equation}
\begin{split}
	\Ext_R^r(\Gamma_*(\cF) , R) &\cong \Gamma_*\left( \SHom_{\P^n}(\cF(-n-1) , \SExt_{\P^n}^r(i_*\cO_X , \omega_{\P^n})) \right) \\
	& \cong \Gamma_*\left( \SHom_{\P^n}(\cF , \SExt_{\P^n}^r(i_*\cO_X , \omega_{\P^n})) (n+1) \right) \\
	& \cong \Gamma_*\left( \SHom_{\P^n}(\cF , \omega^\circ_X) (n+1)\right) \: ,
\end{split}
\end{equation}
which concludes the proof.
\end{proof}
This result establishes that the dualizing module of $\Gamma_*(\cF)$ arises geometrically from the sheaf $\SHom_{\mathcal{O}_{\P^{\scalemath{0.4}{n}}}} (\iota_\ast \mathcal{F}, \omega^\circ_X)(n+1)$. The above theorem, together with~\eqref{extr}, implies the following corollary.
\begin{cor}
	If $\Gamma_*(\cF)$ is a Cohen--Macaulay $R$-module, then we have
	\begin{equation}
		\mu A^\bu(\Gamma_*(\cF))^\vee \cong \mu A^\bu(\Gamma_*(\SHom_{\mathcal{O}_{\P^{\scalemath{0.4}{n}}}} (\iota_\ast \mathcal{F}, \omega^\circ_X) (n+1))) \: . 
	\end{equation}
\end{cor}
It is possible to prove whether or not a sheaf $\cF$ gives rise to a Cohen--Macaulay module via $\Gamma_*$ by studying its sheaf cohomology. In particular, the following result holds~\cite[\S3.8, (3.15.1)]{Kollar}.
\begin{prop}\label{prop: cm}
	Let $X$ be a Cohen-Macaulay projective scheme and let $\mathcal{L}$ be an ample line bundle on it. Given a coherent sheaf $\mathcal{F}$ on $X$, then $\Gamma_\ast (\mathcal{F})$ is a Cohen-Macaulay $R$-module if and only if $H^i (X, \mathcal{F} \otimes \mathcal{L}^{\otimes k}) = 0$ for any $0 < i < \dim (X)$ for any $k \in \mathbb{Z}$.
\end{prop}

\subsection{Duality and line bundles}

As a case study for the general results above, we can thus read off the length of the minimal free resolution by examining the multiplets in~\S3.2 and~\S3.3, recalling that $\mu A^\bullet (-n,0) $ is $\mu A^\bullet (0,n)$ up to a shift.
For a start, recall that the field content of the multiplet $\mu A^\bu(n,0)$ for $n \in \mathbb{Z}$ takes values in the minimal free resolution of the $R$-module $\Gamma_*(\cO_Y(n,0))$. Therefore, given the results in the previous section, we can easily read off $l_R(\Gamma_*(\cO_Y(n,0)))$:
\begin{equation}
l_R(\Gamma_*(\cO(n,0) ) ) = \begin{cases}
3 & \text{for} \ n\in\{-1,0,1,2,3\} \\
4 & \text{for} \ n>3 \\
6 & \text{for} \ n<-2. 
\end{cases}
\end{equation}
Notice that all the modules $\Gamma_\ast (\cO_Y (n,0))$  come from line bundles supported on the nilpotence variety $Y \cong \P^1 \times \P^3 \subset \P^7$, which is of codimension $3$ in $\mathbb{P}^7$. From this, we can infer the following lemma.
\begin{lem} \label{lemmaCM} The $R$-module
	$\Gamma_*(\cO(n,0))$ is Cohen--Macaulay if and only if $n \in \{-1,0,1,2,3\}$.
\end{lem}
\noindent	Therefore, for $n$ in this range, we have
\begin{equation}
\mu A^\bu(n,0)^\vee \cong \mu A^\bu(\Ext^3_R(\Gamma_*(\cO(n,0)),R)) \: .
\end{equation}
Lemma~\ref{lemmaCM} can also be proved directly by studying the sheaf cohomology of the line bundles $\cO(n,m)$ and using Proposition~\ref{prop: cm}. Indeed, we can choose $\cL = \cO(1,1)$ as an ample line bundle and use the Künneth theorem to verify that the middle cohomologies ${H^i(Y,\cO(n+k,k))}$ vanish for $i = 1,2,3$ and for all $k$ precisely when $n \in \{-1,0,1,2,3\}$.

Furthermore, since $Y = \P^1 \times \P^3$, the dualizing sheaf can be described explicitly as the exterior tensor product of the respective dualizing sheaves on the factors. Explicitly,
\begin{equation}
	\omega^\circ_Y = \pi_1^* \omega_{\P^1} \otimes \pi_3^* \omega_{\P^3} = \cO(-2,-4) \: .
\end{equation}
Using this together with Proposition~\ref{duality} gives
\begin{equation}
\begin{split}
	\Ext^3_R(\Gamma_*(\cO(n,0)) , R) &= \Gamma_*(\cO(n,0)^\vee \otimes \cO(-2,-4))(5) \\
	&= \Gamma_*(\cO(-n-2,-4))(5) \\
	&= \Gamma_*(\cO(2-n , 0))(1) \: . 
\end{split}
\end{equation}
Thus we obtain
\begin{equation}
	\mu A^\bu(2-n,0)[1] = \mu A^\bu\left( \Ext^3_R(\Gamma_*(\cO(n,0)) , R) \right) \: .
\end{equation}
In the range where $\Gamma_*(\cO(n,0))$ is Cohen--Macaulay, this implies
\begin{equation} \label{antifn0}
	\mu A^\bu(n,0)^\vee \cong \mu A^\bu(2-n,0)[1]\: .
\end{equation}
This can be viewed as a remnant of Serre duality for line bundles on the multiplet side.

\section{Short exact sequences}
In this section, we discuss some general conclusions that can be drawn about short exact sequences of vector bundles in the context of the pure spinor superfield formalism, and then move on to study some concrete examples in our six-dimensional setting. 
The sections that follow will study the multiplets associated to the tangent and normal bundles and their duals, and apply these results in the context of natural short exact sequences in which those vector bundles appear. 
As such, we are motivated both by abstract considerations---having understood that the pure spinor construction is a functor, it is natural to ask about it not just on single objects, but on diagrams of objects---and, as throughout this paper, by concrete computational examples.

\subsection{General observations}
Let
\begin{equation} \label{eq: seq}
0 \longrightarrow \Gamma' \longrightarrow \Gamma \longrightarrow \Gamma'' \longrightarrow 0
\end{equation}
be a short exact sequence of graded equivariant $R/I$-modules. 
We recall that such a sequence presents $\Gamma$ as an \emph{extension} of $\Gamma''$ by $\Gamma'$. It is standard that equivalence classes of such extensions are classified by $\Ext^1(\Gamma'',\Gamma')$. It is clear that the sequence~\eqref{eq: seq} defines a two-step filtration on the module $\Gamma$, whose associated graded module is $\Gamma' \oplus \Gamma''$.

This motivates the following definition.
\begin{dfn} \label{def: def}
    Let $E$ and $\tilde{E}$ be multiplets. We say that $\tilde E$ is a \emph{deformation} of 
    $E$ if it admits a filtration $F^\bu \tilde E$ with the property that 
    \[
        \Gr F^\bu \tilde E \cong E.
    \]
\end{dfn}
The functor $A^\bullet$ is exact; applying it to the extension sequence~\eqref{eq: seq} produces a short exact sequence of strict multiplets
\begin{equation} \label{eq: multses}
    0 \longrightarrow A^\bullet(\Gamma') \longrightarrow A^\bullet(\Gamma) \longrightarrow A^\bullet(\Gamma'') \longrightarrow 0 .
\end{equation}
It is then immediate to see that the following holds.
\begin{obs} \label{obs}
    Given a short exact sequence of $R/I$-modules as in~\eqref{eq: seq}, the multiplet $A^\bu(\Gamma)$ is a deformation of the direct sum $A^\bu(\Gamma') \oplus A^\bu(\Gamma'')$ in the sense of Definition~\ref{def: def}.
\end{obs}

We will be interested in working with such deformations a bit more concretely; to that end, we record a few elementary observations here.
Since $\Gamma \cong \Gamma' \oplus \Gamma''$ noncanonically as graded vector spaces (even equivariantly for $\fp_0$), the underlying graded super vector bundles $A^\bu(\Gamma)^\#$ and $A^\bu(\Gamma')^\# \oplus A^\bu(\Gamma'')^\#$ of the corresponding multiplets are also isomorphic. The deformation is thus visible only at the level of the differential; the module structure for the supersymmetry algebra is not affected at this level, as everything is freely resolved over the supermanifold $T$.

Recall that the differential on $A^\bu(\Gamma)$ arises from the module structure via the formula 
\begin{equation} \label{eq: ps-diff}
    \cD = \lambda^\alpha \mathpzc{R}(Q_\alpha) \: .
\end{equation}
Making reference to a choice of splitting as vector spaces, the module structure on $\Gamma$ will differ from that on the direct sum by terms in which an element of~$R$ acts by a map from $\Gamma''$ to $\Gamma'$. Equivalence classes of such maps are precisely classes in $\Ext^1(\Gamma'', \Gamma')$.

In practice, we are often interested in component-field descriptions, and therefore specifically in the minimal multiplets $\mu A^\bu(\Gamma)$, $\mu A^\bu(\Gamma')$, and $\mu A^\bu(\Gamma'')$. The additional terms in the differential that are induced by the extension class 
may contain terms of order zero as differential operators, so that there is no way to identify $\mu A^\bu(\Gamma)^\#$ with the direct sum $\mu A^\bu(\Gamma')^\# \oplus \mu A^\bu(\Gamma'')^\#$---even at the level of vector bundles. 
Similarly, the supersymmetry module structures on the component-field models are obtained by homotopy transfer, and thus \emph{are} affected by the deformation. 
At the level of the differential, though, there is a way to present the deformation of multiplets concretely at the component-field level, as we show in the following lemma.
\begin{lem}
    \label{lem: deformation}
	Given a short exact sequence of $R/I$-modules as above, there is a deformation $d'$ of the differential on the direct sum multiplet $\mu A^\bu(\Gamma') \oplus \mu A^\bu(\Gamma'')$ and a roof of quasi-isomorphisms inducing an equivalence 
	\begin{equation}
            \left[ \mu A^\bu(\Gamma) , d_\Gamma \right] \simeq \left[ \mu A^\bu(\Gamma') \oplus \mu A^\bu(\Gamma'') , d_{\Gamma'} + d_{\Gamma''} + d' \right] = \left[ \mu A^\bu(\Gamma' \oplus \Gamma'') , d_{\Gamma' \oplus \Gamma''}  + d'\right].
	\end{equation}
\end{lem}
\begin{proof}
    We start from Observation~\ref{obs}. We can apply the Rees construction~\cite{Rees} to construct a family  
	\begin{equation}
        \left( A^\bu(\Gamma' \oplus \Gamma'')[t], \cD_{\Gamma' \oplus \Gamma''} + t\cD') \right)
	\end{equation}
        over the disk, whose fiber at $t=1$ is identified with $\left(A^\bu(\Gamma) ,\cD_{\Gamma} \right)$. 
        We can now apply the spectral sequences associated to two different filtrations. First, we can filter as described in~\S\ref{ssec: construct} and take cohomology with respect to the Koszul differential. The corresponding homotopy transfer yields the left-hand-side $\mu A^\bu(\Gamma)$. Second, we can apply the same filtration, but now also assigning weight one to the deformation parameter $t$. Correspondingly, we first take cohomology with respect to the Koszul differential associated to the direct sum module structure on $\Gamma' \oplus \Gamma''$. Hence, the cohomology is precisely $\mu A^\bu(\Gamma' \oplus \Gamma'')$. Performing the homotopy transfer induces additional pieces coming from the deformation piece of the differential that constitute the new deformation of the differential at component field level.
        Specializing to the fiber at $t=1$, we obtain the statement.
\end{proof}
In this work we often deal with short exact sequences of equivariant sheaves on $Y$. Let
\begin{equation}
0 \longrightarrow \cF' \longrightarrow \cF \longrightarrow \cF'' \longrightarrow 0
\end{equation}
be such a sequence. First, we observe that taking the tensor product with a line bundle keeps the sequence exact, thus we obtain short exact sequences
\begin{equation}
0 \longrightarrow \cF'(k) \longrightarrow \cF(k) \longrightarrow \cF''(k) \longrightarrow 0
\end{equation}
for all $k \in \Z$. Second, a short exact sequence induces a long exact sequence in cohomology
\begin{equation}
0 \longrightarrow H^0(\cF'(k)) \longrightarrow H^0(\cF(k)) \longrightarrow H^0(\cF''(k)) \xlongrightarrow{\delta} H^1(\cF'(k)) \longrightarrow \dots
\end{equation}
Thus, if the map $\delta$ vanishes (for example due to $H^1(\cF'(k))$ being zero) for all $k$, we obtain a short exact sequence on the global sections
\begin{equation}
0 \longrightarrow H^0(\cF'(k)) \longrightarrow H^0(\cF(k)) \longrightarrow H^0(\cF''(k)) \longrightarrow 0 \: ,
\end{equation}
and therefore a short exact sequence of graded equivariant $R/I$-modules
\begin{equation}
0 \longrightarrow \Gamma_*(\cF') \longrightarrow \Gamma_*(\cF) \longrightarrow \Gamma_*(\cF'') \longrightarrow 0 \: .
\end{equation}
and we find ourselves in the situation described above.

\begin{rmk}
    It is worth recalling that extensions of sheaves are often interpreted as related to interactions or bound states in mathematical physics. Just for example, in topological string theory, $B$-branes are identified with coherent sheaves on the target space, which is typically a Calabi--Yau threefold. As emphasized in early work on the subject~\cite[for example]{Sharpe,Douglas,Aspinwall}, a nontrivial extension sequence of the form 
    \[
        0 \to A \to B \to C \to 0
    \]
    indicates that $B$ should be thought of as a bound state of the branes $A$ and~$C$. (Making this interpretation precise led to the identification of the category of $B$-branes with the \emph{derived} category of coherent sheaves.)

    In our setting, as explained, the extension defines a deformation of the module structure, which in turn deforms the differential on the multiplet. Thinking in the context of the Batalin--Vilkovisky formalism, a deformation of the differential can in turn be thought of as a deformation of the quadratic part of the BV action. As such, the new differentials we consider on component fields can be interpreted, at least schematically, as (quadratic) supersymmetric interactions between the multiplets $\mu A^\bu(\Gamma')$ and~$\mu A^\bu(\Gamma'')$, such that the deformed multiplet has derived supertranslation invariants~$\Gamma$.
\end{rmk}

In the following, we will study several examples of short exact sequences and the associated deformations of multiplets. \S\ref{sec:comp} will examine a construction that produces the vector multiplet in three-dimensional $\N=1$ supersymmetry from a deformation of scalar- and spinor-valued free superfields, and compute the deformation of the differential $d'$ explicitly, using the procedure of Lemma~\ref{lem: deformation}. In other examples, we state the form of $d'$ without proof, or just observe explicitly that differentials of the appropriate form exist (though this is of course guaranteed by the general results above).

\subsection{Example: three-dimensional $\cN=1$}
\label{sec:comp}
Let us illustrate the findings from above in the case of three-dimensional $\cN=1$ supersymmetry. 
We will give an example of a short exact sequence of $R/I$-modules, and compute the deformation $d'$ from Lemma~\ref{lem: deformation} explicitly.

In three dimensions we have $\mathfrak{so}(3) \cong \mathfrak{sl}(2)$ as complex Lie algebras. We denote the two-dimensional spinor representation by $S$ and the three-dimensional vector representation by $V$. The super Poincar\'e algebra takes the form
\begin{equation}
    \fp = \mathfrak{sl}(2) \oplus S(-1) \oplus V(-2),
\end{equation}
with the bracket of odd elements given by the isomorphism $\Gamma: \sym^2(S) \cong V$. Writing $R = \C[\lambda^\alpha]$ with $\alpha=1,2$, the defining ideal of the nilpotence variety is
\begin{equation}
	I = ( (\lambda^1)^2 , \lambda^1 \lambda^2 , (\lambda^2)^2 ) = R_{\geq 2}.
\end{equation}
Thus $R/I$ is a finite-dimensional algebra over~$C$, supported in weights zero and one. It is the square-zero extension of the ground field $\C$ by the vector space $S$.
By a slight abuse of notation, we will just write $\C$ for the corresponding $R/I$-module (the quotient by the maximal ideal), and similarly $S$ for the $R/I$-module structure on the vector space arising from the quotient map $R/I \to \C$.

\begin{prop}
Consider the short exact sequence of $R/I$-modules 
\begin{equation}
    \label{eq:ses}
    0 \longrightarrow S(-1) \longrightarrow R/I \longrightarrow \C \longrightarrow 0.
\end{equation}
We recall that $\mu A^\bu(\C) = C^\infty(T)$. 
The differential $d'$ induced by the extension, via the construction in Lemma~\ref{lem: deformation}, is the sum of the following four maps from $C^\infty(T)$ to $C^\infty(T) \otimes_\C S$:
\begin{itemize}
    \item the identity morphism from $S[1](-1) \subset C^\infty(T)$ to $S(-1) \subset C^\infty(T) \otimes_\C S(-1)$;
    \item the equivariant inclusion map from $\Omega^0[2](-2) \subset C^\infty(T)$ to $\Omega^0[1](-2) \subset S[1](-1) \otimes S(-1) \subset C^\infty(T) \otimes_\C S(-1)$;
    \item the de Rham differential from $\Omega^0 \subset C^\infty(T)$ to $\Omega^1[1](-2) \subset S[1](-1) \otimes S[-1] \subset C^\infty(T) 
        \otimes_\C S(-1)$;
    \item the Dirac operator from $S[1](-1) \subset C^\infty(T)$ to $S[2](-3) \subset C^\infty(T) \otimes_\C S(-1)$.
\end{itemize}
We summarize these maps with the arrows in the following table:
\begin{equation}
\left[
\begin{tikzcd}[row sep=0.5cm, column sep=0.4cm]
    \textcolor{blue!65!black}{\Omega^0} \arrow[dr, "\d", near start] & \textcolor{blue!65!black}{S} \arrow[dl, crossing over, anchor=top] \arrow[dr, "\slashed{\partial}", near start] & \textcolor{blue!65!black}{\Omega^0} \arrow[dl, crossing over] \\
    \textcolor{green!65!black}{S} & \textcolor{green!65!black}{\smash{\Omega^1 \oplus \Omega^0}\vphantom{\Omega\oplus\Omega^0}} & \textcolor{green!65!black}{S}
\end{tikzcd} \right].
\end{equation}
\end{prop}
\begin{proof}
    The extension~\eqref{eq:ses} is nontrivial. The module structure on the direct sum $\C \oplus S(-1)$ is inherited from the vector space structure; the deformation of the module structure to that on $R/I$ is given by the map
\begin{equation} \label{eq: deform mod}
	R/I \times \C \longrightarrow S(-1) \qquad (\lambda^\alpha , 1) \mapsto s^\alpha \: . 
\end{equation}
The multiplet $A^\bu(\C)$ is a free superfield: it takes the form 
\[
    C^\infty(T) = C^\infty(\R^3) \otimes \wedge^\bu S^\vee(1),
\]
generated by variables $\theta^\alpha$ of cohomological degree zero and internal degree $-1$. 
Similarly, $A^\bu(S(-1))$ is a free superfield with values in $S(-1)$; we write $s^\alpha$ for a basis for this copy of $S$. The differentials $\cD_{\C}$ and~$\cD_{S(-1)}$ are zero, and so $\mu A^\bu(\C) \cong A^\bu(\C)$.

The representation $S$ is self-dual. Using the isomorphism $S \otimes S \cong V \oplus \C$, and denoting fields in the vector representation by  $\Omega^1$ and fields in the trivial representation by $\Omega^0$, the direct sum $\mu A^\bu(\C) \oplus \mu A^\bu(S(-1))$ takes the form 
\begin{equation}
 \textcolor{blue!65!black}{\mu A^\bu(\C)} \oplus 
    \textcolor{green!65!black}{\mu A^\bu(S(-1))} = \left[
\begin{tikzcd}[row sep=0.3cm, column sep=0.2cm]
\textcolor{blue!65!black}{\Omega^0} & \textcolor{blue!65!black}{S} & \textcolor{blue!65!black}{\Omega^0} \\
\textcolor{green!65!black}{S} & \textcolor{green!65!black}{\Omega^1 \oplus \Omega^0} & \textcolor{green!65!black}{S}
\end{tikzcd} \right]
\end{equation}
It is straightforward to see that the component fields have the following representatives:
\begin{equation}
	\begin{bmatrix}
	1 & \theta^\alpha & \theta^1 \theta^2 \\
	s^\alpha & \theta^\alpha s^\beta & \theta^1 \theta^2 s^\alpha 
	\end{bmatrix} \: .
\end{equation}
Since the differential is zero and $R/I = \C \oplus S(-1)$ as vector spaces, 
we have an immediate identification 
\deq{
    \mu A^\bu(\C) \oplus \mu A^\bu(S(-1)) \cong A^\bu(R/I)^\#,
}
as well as an identification between the differentials $d'$ and $\cD_{R/I}$. 
Applying the general formula
\begin{equation}
	\cD = \lambda^\alpha \frac{\partial}{\partial \theta^\alpha} + \lambda^\alpha \theta^\beta \Gamma^\mu_{\alpha \beta} \frac{\partial}{\partial x^\mu} \: ,
\end{equation}
and using the representatives, 
we recover the four terms in the differential by a straightforward direct computation. 
Using the other filtration in Lemma~\ref{lem: deformation} just recovers the component-field spectral sequence, giving us the component-field model $\mu A^\bu(R/I)$.
\end{proof}

\subsection{The Euler sequence for $\P^1$}
Let us now discuss a family of short exact sequences in six dimensions. Identifying $\cT_{\P^1} \cong \cO_{\P^1}(2)$, the Euler exact sequence for $\P^1$ reads
\begin{equation}
0 \longrightarrow \cO_{\P^1} \longrightarrow \cO_{\P^1}(1) \otimes \C^2 \longrightarrow \cO_{\P^1}(2) \longrightarrow 0 \: .
\end{equation}
Note that this is a sequence of equivariant sheaves and that $\C^2$ carries the fundamental representation of $\mathfrak{sl}(2)$. Twisting by $\cO_{\P^1}(n)$ and pulling back along $\pi_1$ we obtain a family of short exact sequences of equivariant sheaves on $Y$
\begin{equation}
0 \longrightarrow \cO(n,0) \longrightarrow \cO(n+1,0) \otimes \C^2 \longrightarrow \cO(n+2,0) \longrightarrow 0 \: . 
\end{equation}
Let us restrict to the case $n\geq 0$ for the moment. Twisting by $\cO_Y(k) = \cO(k,k)$ we obtain the sequences
\begin{equation}
0 \longrightarrow \cO(n+k,k) \longrightarrow \cO(n+k+1,k) \otimes \C^2 \longrightarrow \cO(n+k+2,k) \longrightarrow 0 \: .
\end{equation}
The relevant first cohomology group is $H^1(\cO(n+k,k))$ which is easily seen to vanish for all $k \in \Z$ by the K\"unneth theorem. Thus, we obtain for all $n\geq0$ a short exact sequence of graded equivariant $R/I$-modules,
\begin{equation}
0 \longrightarrow \Gamma_*(\cO(n,0)) \longrightarrow \Gamma_*(\cO(n+1,0)) \otimes \C^2 \longrightarrow \Gamma_*(\cO(n+2,0)) \longrightarrow 0 \: .
\end{equation}
Let us study the associated multiplets.

\subsubsection{$n=0$}
Recall that $\mu A^\bu(0,0)$ is the vector multiplet, $\mu A^\bu(1,0)$ the hypermultiplet and $\mu A^\bu(2,0)[1] = \mu A^\bu(0,0)^\vee$ the antifield multiplet of the vector, according to \eqref{antifn0}. Therefore (up to shifting), $\mu A^\bu(\cO(1,0) \otimes \C^2) = \mu A^\bu(\cO(1,0)) \otimes \C^2$ is a doublet of hypermultiplets with values in the fundamental representation of the $R$-symmetry $\mathfrak{sl}(2)$. Let us arrange the direct sum as follows.
\begin{equation}
\textcolor{green!65!black}{\mu A^\bu(0,0)^\#} \oplus \textcolor{blue!65!black}{\mu A^\bu(2,0)^\#} = \left[
\begin{tikzcd}[row sep=0.3cm, column sep=0.2cm]
\textcolor{green!65!black}{\C} \oplus \textcolor{blue!65!black}{\C^3} & \textcolor{blue!65!black}{S_+ \otimes \C^2} & \textcolor{blue!65!black}{\Omega^1}\\
& \textcolor{green!65!black}{\Omega^1} & \textcolor{green!65!black}{S_- \otimes \C^2} & \textcolor{blue!65!black}{\C} \oplus \textcolor{green!65!black}{\C^3}
\end{tikzcd} \right]
\end{equation}
We can deform it by adding an acyclic differential relating the the two one-forms, the Dirac operator for the fermions, and the Laplacian for the scalar fields.  
\begin{equation}
\left[\textcolor{green!65!black}{\mu A^\bu(0,0)} \oplus \textcolor{blue!65!black}{\mu A^\bu(2,0)} \right]^{\mathrm{Deform}} = \left[
\begin{tikzcd}[row sep=0.5cm, column sep=0.5cm]
    \textcolor{green!65!black}{\C} \oplus \textcolor{blue!65!black}{\C^3} \arrow[drrr, "\star d \star d" ' near start] & \textcolor{blue!65!black}{S_+ \otimes \C^2} \arrow[dr, "\slashed{\partial}" , near start ] & \textcolor{blue!65!black}{\Omega^1} \arrow[dl, swap,  "\id" ' near end, crossing over]\\
& \textcolor{green!65!black}{\Omega^1} & \textcolor{green!65!black}{S_- \otimes \C^2} & \textcolor{blue!65!black}{\C} \oplus \textcolor{green!65!black}{\C^3}
\end{tikzcd} \right]
\end{equation}
Taking cohomology with respect to the acyclic part of the differential (i.e. integrating out the auxiliary field) and recalling that for $\mathfrak{sl}(2)$-representations $\C^2 \otimes \C^2 \cong \C \oplus \C^3$, we immediately see that we recover the hypermultiplet with values in $\C^2$.

Interestingly there is another BV theory which can be formed out of $\mu A^\bu(0,0)$ and $\mu A^\bu(0,2)$. Adding both multiplets with an appropriate shift and deforming the resulting complex one obtains the BV theory of the vector multiplet~\cite[\S4.3]{perspectives}. Denoting the vector multiplet by $E$, these findings can be summarized by stating that the cotangent theory $T^\vee[-1]E$ corresponds to the BV theory describing the vector multiplet, while the construction we presented above corresponds to $(T^\vee[1]E)[-1]$ which is seen to be equivalent to the hypermultiplet. Let us finally remark that all these considerations are purely perturbative.
\subsubsection{$n=1$}
Proceeding analogously, we can define a deformation on the direct sum of $\mu A^\bu(1,0)$ and~$\mu A^\bu(3,0)$ that renders it quasi-isomorphic to~$\mu A^\bu(2,0) \otimes \C^2$:
\begin{gather}
\left[\textcolor{green!65!black}{\mu A^\bu(1,0)} \oplus \textcolor{blue!65!black}{\mu A^\bu(3,0)} \right]^{\mathrm{Deform}}= \notag \\
\left[
\begin{tikzcd}[row sep=0.3cm, column sep=0.2cm, ampersand replacement = \&]
\textcolor{green!65!black}{\C^2} \oplus \textcolor{blue!65!black}{\C^4} \& S_+ \otimes (\textcolor{green!65!black}{\C} \oplus \textcolor{blue!65!black}{\C^3}) \& \textcolor{blue!65!black}{\wedge^2 S_+ \otimes \C^2} \arrow[dr, "\star d \star" 'near end]
\& \textcolor{blue!65!black}{S_-} \arrow[dl , "\mathrm{id}" ' near start, crossing over] \\
\& \& \textcolor{green!65!black}{S_-} \& \textcolor{green!65!black}{\C^2}
\end{tikzcd} \right] \\
\simeq \mu A^\bu(2,0) \otimes \C^2. \notag
\end{gather}
Here we used the isomorphisms $V \cong \wedge^2 S_+$ and $S_- \cong \wedge^3 S_+$.

\subsubsection{$n=2$}
Similarly, there is a deformation of $\mu A^\bu(2,0) \oplus \mu A^\bu(4,0)$ giving $\mu A^\bu(3,0)\otimes \C^2$.
\begin{gather}
[\textcolor{green!65!black}{\mu A^\bu(2,0)} \oplus \textcolor{blue!65!black}{\mu A^\bu(4,0)}]^{\mathrm{Deform}} = \notag \\
\left[
\begin{tikzcd}[row sep=0.3cm, column sep=0.2cm, ampersand replacement = \&]
\textcolor{green!65!black}{\C^3} \oplus \textcolor{blue!65!black}{\C^5} \& S_+ \otimes(\textcolor{green!65!black}{\C^2} \oplus \textcolor{blue!65!black}{\C^4})
\& \wedge^2 S_+ \otimes(\textcolor{green!65!black}{\C} \oplus \textcolor{blue!65!black}{\C^3}) \& \textcolor{blue!65!black}{\wedge^3 S_+ \otimes \C^2} \& \textcolor{blue!65!black}{\C} \arrow[dl,"\mathrm{id}"] \\
\& \& \& \textcolor{green!65!black}{\C}
\end{tikzcd} \right] \\
\simeq \mu A^\bu(3,0) \otimes \C^2. \notag
\end{gather}

\subsubsection{$n\geq 3$}
For $n \geq 3$ , we interpreted $\mu A^\bu(n,0)$ as receiving an $n$-fold covering map from the hypermultiplet, witnessed by the isomorphism between its observables and the subalgebra of hypermultiplet observables with polynomial degree divisible by~$n$. The short exact sequence gives a relation between $\mu A^\bu(n+1,0) \otimes \C^2$ and $\mu A^\bu(n,0) \oplus \mu A^\bu(n+2,0)$:
\begin{gather}
\textcolor{green!65!black}{\mu A^\bu(n,0)} \oplus \textcolor{blue!65!black}{\mu A^\bu(n+2,0)} = \\  
\left[
\begin{tikzcd}[row sep=0.3cm, column sep=0.1cm, ampersand replacement = \&]
\textcolor{green!65!black}{\C^{n+1}} \oplus \textcolor{blue!65!black}{\C^{n+3}} \& (\textcolor{green!65!black}{\C^{n}} \oplus \textcolor{blue!65!black}{\C^{n+2}}) \otimes S_+ \& (\textcolor{green!65!black}{\C^{n-1}} \oplus \textcolor{blue!65!black}{\C^{n+1}}) \otimes \wedge^2 S_+ \& (\textcolor{green!65!black}{\C^{n-2}} \oplus \textcolor{blue!65!black}{\C^{n}}) \otimes \wedge^3 S_+ \& \textcolor{green!65!black}{\C^{n-3}} \oplus \textcolor{blue!65!black}{\C^{n-1}}
\end{tikzcd} \right] \notag \\
\simeq \mu A^\bu(n+1,0) \otimes \C^2. \notag
\end{gather}
Here, identifying $\mu A^\bu(n+1) \otimes \C^2$ just amounts to the decomposition rule $\C^n \otimes \C^2 \cong \C^{n+1} \oplus \C^{n-1}$. In other words, considering pairs of hypermultiplet observables of polynomial degree $(n+1)$ and regarding such pairs as transforming in the fundamental representation of the $R$-symmetry $\mathfrak{sl}(2)$, we can either symmetrize with respect to the $R$-symmetry index (yielding $\mu A^\bu(n+2,0)$) or antisymmetrize to land in $\mu A^\bu(n,0)$. Note, however, that the polynomial degree of the observables involved does \emph{not} change; the multiplet $\mu A^\bu(n,0)$, rather than its realization via a map from the hypermultiplet, is the fundamental object.

\section{The normal bundle exact sequence} \label{sec: normal}
In the last two sections, we use short exact sequences for geometric bundles on $\P^1 \times \P^3$ to extend our survey of multiplets to higher-rank bundles. In this section, we treat the tangent and normal bundles, using the defining short exact sequence that relates them; in the following section, we will consider the dual of this sequence and work out the multiplets involved explicitly.

As recalled above, the tangent bundle of the projectivized nilpotence variety $\cT_Y$, the restriction $\cT_{\P^7}|_Y$ of the tangent bundle of the ambient $\P^7$ to $Y$, and the normal bundle $\cN_Y$ sit in the normal bundle exact sequence
\begin{equation} \label{normalses}
0 \longrightarrow \cT_Y \longrightarrow \cT_{\P^7}|_Y \longrightarrow \cN_Y \longrightarrow 0 \: .
\end{equation}
Since $H^1(\cT_Y(k))=0$ for all $k \in \ZZ$, this short exact sequence induces a short exact sequence on global sections. Thus, applying $\Gamma_*$, we obtain a short exact sequence of $R/I$-modules.
\begin{equation}
0 \longrightarrow \Gamma_*(\cT_Y) \longrightarrow \Gamma_*(\cT_{\P^7}|_Y) \longrightarrow \Gamma_*(\cN_Y) \longrightarrow 0 
\end{equation}
We apply this short exact sequence to study the associated multiplets and their relations to one another. Again, we will find that there is a deformation of $\mu A^\bu(\cT_Y) \oplus \mu A^\bu(\cN_Y)$ which is quasi-isomorphic to $\mu A^\bu(\cT_{\P^7}|_Y)$.

\subsection{Tangent bundle}
\subsubsection{Cohomology and Hilbert Series}
Recall that, as seen above, the tangent bundle to the nilpotence variety is given by the exterior sum
\begin{equation}
\mathcal{T}_Y=  \pi_1^* \mathcal{T}_{\P^1} \oplus \pi_3^* \mathcal{T}_{\P^3} =  \mathcal{T}_{\P^1} \boxplus \mathcal{T}_{\P^3},
\end{equation}
where $\pi^\ast_1 \mathcal{T}_{\proj 1} \cong \mathcal{O}_Y (2,0).$ Accordingly, the resulting multiplet will be given by a direct sum,
\begin{equation}
\mu A^\bullet(\mathcal{T}_Y) = \mu A^\bullet(\cO_Y(2,0)) \oplus \mu A^\bullet(\pi_3^*\mathcal{T}_{\P^3}) .
\end{equation}
We have already identified $\mu A^\bullet(\cO_Y(2,0))$ as the antifield multiplet of the vector multiplet in~\S\ref{sec:(n,0)}, so we are left with studying $\mu A^\bullet(\pi_3^*\mathcal{T}_{\P^3})$, which amounts to computing the zeroth cohomology of
\begin{equation}
\pi_3^* \mathcal{T}_{\P^3} (k) = \pi_1^\ast \mathcal{O}_{\proj 1} (k) \otimes \pi^\ast_3 \mathcal{T}_{\proj 3} (k) = \cO_{\P^1}(k) \boxtimes \mathcal{T}_{\P^3} (k) .
\end{equation}
By the K\"unneth theorem, we have
\begin{equation}\label{eq: Kunneth P3}
H^0(\pi_3^* \mathcal{T}_{\P^3}(k)) = H^0(\cO_{\P^1}(k)) \otimes H^0(\mathcal{T}_{\P^3}(k)).
\end{equation}
which in turn reduces the problem to compute $H^0(\mathcal{T}_{\P^3} (k))$. Twisting the Euler exact sequence \eqref{euler} by $\mathcal{O}_{\P^3} (k)$, we find
\begin{equation} 
0 \longrightarrow \cO_{\P^3}(k) \longrightarrow \cO_{\P^3}(k+1) \otimes S_- \longrightarrow \cT_{\P^3}(k) \longrightarrow 0 \: ,
\end{equation}
where $S_- \cong \mathbb{C}^4$. The long cohomology sequence, for the relevant cases $k\geq 0$, reduces to the following short exact sequence
\bear \label{EuSe}
\xymatrix{
	0 \ar[r] & H^0(\mathcal{O}_{\proj 3} (k) )  \ar[r] & H^0 (\mathcal{O}_{\proj 3} (k+1) ) \otimes \mathbb{C}^{4} \ar[r] & H^0 (\mathcal{T}_{\proj 3} (k) ) \ar[r] & 0,
}
\eear
since $H^1 (\mathcal{O}_{\P^3} (k)) = 0$ for any $k\geq 0$. As a consequence, one has 
\begin{multline}
h^0(\mathcal{T}_{\P^3}(k)) = 4 h^0(\cO_{\P^3}(k+1)) - h^0(\cO_{\P^3}(k))  
= 4 \binom{k+4}{3} - \binom{k+3}{3} 
= \frac{1}{2}(k+2)(k+3)(k+5) .
\end{multline}
The resulting Hilbert series is easily resummed, giving 
\begin{equation}
\begin{split}
\Hilb(\pi_3^*\mathcal{T}_{\P^3}) &= \sum_{k = 0}^{\infty} \frac{1}{2} (k+1)(k+2)(k+3)(k+5)t^k \\
&= \frac{15 - 48t + 54t^2 - 24t^3 + 3t^4}{(1-t)^8} ,
\end{split}
\end{equation}
such that the Betti numbers of the associated multiplet are
\begin{equation}
\grdim\mu A^\bu(\pi_3^\ast \mathcal{T}_{\P^3}) = 	\begin{bmatrix}
15 & 48 & 54 & 24 & 3 \\
\end{bmatrix} \: .
\end{equation}

\subsubsection{Equivariant decomposition}
Since the induced sequence of representations of~\eqref{EuSe} splits, we find in terms of representations of $\mathfrak{sl}(2) \times \mathfrak{sl}(4)$
\begin{equation}
H^0(\cT_{\P^3}(k)) = [0|k+1,0,0] \otimes S_- - [0|k,0,0] = [0|k+1,0,1] \: ,
\end{equation}
and hence by~\eqref{eq: Kunneth P3}
\begin{equation}
H^0((\pi_3^*\cT_{\P^3})(k)) = [k|k+1,0,1] \: .
\end{equation}
Running our machinery, we obtain the representations
\begin{equation}
    \begin{split}
        W_0 &= \phantom{-}[0|1,0,1] \\
        W_1 &= -[1|0,1,1] -[1|1,0,0] \\
        W_2 &= \phantom{-}[0|0,1,0] + [2|0,0,2] + [2|0,1,0] \\
        W_3 &= -[1|0,0,1] - [3|0,0,1] \\
        W_4 &= \phantom{-}[2|0,0,0].
    \end{split}
\end{equation}
Let us summarize the field content.
\begin{equation}
\mu A^\bu(\pi_3^*\cT_{\P^3})^\# = \left[
\begin{tikzcd}[row sep=0.3cm, column sep=0.2cm]
\Omega^2 & S_- \otimes V \otimes \C^2 & V \oplus \Omega^3_- \otimes \C^3 \oplus V \otimes \C^3  & (\C^2 \oplus \C^4) \otimes S_- & \C^3
\end{tikzcd} \right]
\end{equation}

\subsection{Restriction of $\mathcal{T}_{\P^7}$ to the nilpotence variety}
\subsubsection{Cohomology and Hilbert series}
Since restriction to a smooth subvariety is an exact functor, it is easy to describe the vector bundle $\mathcal{T}_{\proj 7}|_Y (k)$ as the quotient bundle sitting in the restriction of the ordinary $k$-twisted Euler exact sequence of the embedding space $\proj 7$ of $Y$, i.e.
\begin{equation}
0 \longrightarrow \cO_Y(k,k) \longrightarrow \cO_Y(k+1,k+1) \otimes [1|0,0,1] \longrightarrow \mathcal{T}_{\P^7}|_Y(k) \longrightarrow 0 \: ,
\end{equation}
where we have used that $\mathcal{O}_{\P^7}|_Y (k) \cong \mathcal{O}_Y (k,k)$. Observing that $H^1 (\mathcal{O}_Y (k,k))$ vanishes for any $k\geq 0$, we find the short exact sequence in cohomology 
\begin{equation} \label{eq: Restriction Euler}
	0 \longrightarrow H^0 (\cO_Y(k,k)) \longrightarrow H^0 (\cO_Y(k+1,k+1)) \otimes [1|0,0,1] \longrightarrow H^0 (\mathcal{T}_{\P^7}|_Y(k)) \longrightarrow 0.
\end{equation}
This yields the formula 
\begin{equation}
\begin{split}
h^0(\mathcal{T}_{\P^7}|_Y(k)) &= 8(k+2)\binom{k+4}{3} - (k+1)\binom{k+3}{3}  \\
&= \frac{4}{3}(k+2)(k+4)(k+3)(k+2) - \frac{1}{6}(k+1)(k+3)(k+2)(k+1)
\end{split}
\end{equation}
for the dimensions of the spaces of global sections of $\mathcal{T}_{\proj 7}|_Y (k)$.  
Notice that this also accounts for the special case $k = -1$, when $H^0 (\mathcal{T}_{\proj 7}|_Y (-1)) \cong H^0 (\mathcal{O}_Y) \otimes [1|0,0,1] \cong [1|0,0,1]$.
The Hilbert series of $\mathcal{T}_{\proj 7}|_Y$ is found to be
\begin{equation}
\Hilb(\mathcal{T}_{\P^7}|_Y) = \frac{8-t-48t^2+70t^3-32t^4+3t^5}{t(1-t)^8}.
\end{equation}

\subsubsection{Equivariant decomposition}
Since~\eqref{eq: Restriction Euler} splits, we find on the level of representations
\begin{equation}
H^0(\cT_{\P^7}|_Y(k)) = [k+1|k+1,0,0] \otimes [1|0,0,1] - [k|k,0,0] \: . 
\end{equation}
The associated representations appearing in $\mu A^\bu(\cT_{\P^7}|_Y)$ are
\begin{equation}
\begin{gathered}
    W_0 = [1|0,0,1], \qquad
    W_1 = -[0|0,0,0], \qquad
    W_2 = -[1|0,1,1] - [1|1,0,0], \\
    W_3 = [0|0,0,2] + [2|0,0,2] + 2[0|0,1,0] + [2|0,1,0], \\
    W_4 = -2[1|0,0,1] - [3|0,0,1], \qquad
    W_5 = [2|0,0,0].
\end{gathered}
\end{equation} 
Explicitly, the field content of $\mu A^\bu(\cT_{\P^7}|_Y)$ is summarized in the array
\begin{gather}
\mu A^\bu(\cT_{\P^7}|_Y)^\# = \notag \\
\left[
\begin{tikzcd}[row sep=0.3cm, column sep=0.2cm, ampersand replacement = \&]
S_- \otimes \C^2 \& \Omega^0 \\
\& S_- \otimes V \otimes \C^2 \& \smash{\begin{array}{c}
\Omega^1 \otimes(\C \oplus \C \oplus \C^3) \\ \Omega^3_- \otimes (\C \oplus \C^3)
\end{array}} \& S_- \otimes (\C^2 \oplus \C^2 \oplus \C^4) \& \C^3 
\end{tikzcd} \right] . 
\end{gather}
Note that, here and in the following---as the multiplets get bigger---, we allow for linebreaks within entries of the table (here in degree (3,0)); the direct sum sign between the representations in the respective signs is left understood. 
This multiplet looks like a gravitino multiplet, containing a spin-${3}/{2}$ (Rarita--Schwinger) field, but no metric or other degree of freedom corresponding to a particle of spin two.

\subsection{Normal bundle}
\subsubsection{Cohomology and Hilbert series}
Using our results on the cohomology of the bundles $\mathcal{T}_{\proj 7}|_Y (k)$ and $\mathcal{T}_Y (k)$, it is easy to compute the cohomology of $\mathcal{N}_{Y} (k)$ by means of the twisted normal bundle exact sequence
\begin{equation}
0 \longrightarrow \mathcal{T}_Y(k) \longrightarrow \mathcal{T}_{\P^7}|_Y(k) \longrightarrow \cN_Y(k) \longrightarrow 0 \: .
\end{equation}
Since $H^1 (\mathcal{T}_Y (k)) = 0$ for any $k$, as can be seen upon using K\"unneth theorem in combination with the twisted Euler exact sequence to evaluate the first cohomology group of $\mathcal{T}_{\proj 3} (k)$, then one finds a short exact sequence in cohomology
\bear
0 \longrightarrow H^0 (\mathcal{T}_Y(k)) \longrightarrow H^0( \mathcal{T}_{\P^7 }|_Y(k)) \longrightarrow H^0 (\cN_Y(k)) \longrightarrow 0 .
\eear
This implies that for $k \geq - 1$ one finds the Betti numbers
\begin{equation}
h^0(\cN_Y(k)) = h^0(\mathcal{T}_{\P^7}|_Y(k)) - h^0(\mathcal{T}_Y(k)) \: .
\end{equation}
Using our previous results for $h^0(\mathcal{T}_{\P^7}|_Y(k))$ and $h^0(\mathcal{T}_Y(k))$ we can deduce
\begin{equation}
h^0(\cN_Y(k)) = \frac{1}{2}(k+3)^2(k+2)(k+5) \: .
\end{equation}
Finally, resumming the Hilbert series yields
\begin{equation}
\Hilb(\cN_Y) = \frac{8 - 19t + 8t^2 +10t^3 -8t^4 + t^5}{t(1-t)^8} \: .
\end{equation}
\subsubsection{Equivariant decomposition}
For the equivariant decomposition, we identify
\begin{equation}
H^0(\cN_Y(k)) = [k+2|k+1,0,1] \: .
\end{equation}
Using this, we can find the representations appearing in~$\mu A^\bu(\cN_Y)$:
\begin{equation}
\begin{gathered}
    W_0 = [1|0,0,1], \qquad
    W_1 = -[0|0,0,0] - [0|1,0,1] - [2|0,0,0], \\
    W_2 = [1|1,0,0], \qquad
    W_3 = [0|0,0,2], \qquad
    W_4 = -[1|0,0,1], \qquad
    W_5 = [0|0,0,0].
\end{gathered}
\end{equation}
Explicitly, the multiplet takes the following form:
\begin{gather}
\mu A^\bu(\cN_Y)^\# = \notag \\
\left[
    \begin{tikzcd}[row sep=0.3cm, column sep=0.2cm, ampersand replacement = \&]
S_- \otimes \C^2 \& \C \oplus \Omega^2 \oplus \C^3 \& S_+ \otimes \C^2 \\
\& \&  \Omega^3_- \& S_- \otimes \C^2 \& \C 
\end{tikzcd} \right] .
\end{gather}

\subsection{Deformation}
As in the example of the Euler sequence for $\P^1$, we find that the field contents of the direct sums of the multiplets associated to tangent and normal bundle does not match the field content of $\mu A^\bu(\cT_{\P^7}|_Y)$, i.e.
\begin{equation}
\mu A^\bu(\cT_Y)^\#\oplus \mu A^\bu(\cN_Y)^\# \neq \mu A^\bu(\cT_{\P^7}|_Y)^\#  .
\end{equation}
This is again related to the fact that the normal exact sequence does not split as a sequence of $R/I$-modules. However, there is a deformation of the direct sum such that the resulting multiplet recovers the field content of $\mu A^\bu(\cT_{\P^7}|_Y)$ up to quasi-isomorphism:
\begin{gather}
\mu A^\bu(\cT_{\P^7}|_Y) \simeq 
[\textcolor{green!65!black}{\mu A^\bu(\cT_Y)^\#} \oplus \textcolor{blue!65!black}{\mu A^\bu(\cN_Y)^\#}]^{\mathrm{Deform}} = \notag \\
\left[
\begin{tikzcd}[row sep=0.3cm, column sep=0.2cm, ampersand replacement = \&]
    \textcolor{blue!65!black}{S_- \otimes \C^2} \&   \textcolor{blue!65!black}{\C \oplus} \circled{\textcolor{blue!65!black}{ \Omega^2 \oplus \C^3}} \arrow[dl,"\id", swap]
\& \circled{\textcolor{blue!65!black}{S_+ \otimes \C^2}} \arrow[dl,"\id", swap] \\ 
\circled{\textcolor{green!65!black}{\Omega^2 \oplus \C^3}}
\& \begin{split}
&\circled{\textcolor{green!65!black}{S_+ \otimes \C^2}} \\ &\textcolor{green!65!black}{ S_- \otimes V \otimes \C^2} 
\end{split}
\& \begin{split}
&\Omega^3_- \otimes (\textcolor{blue!65!black}{\C} \oplus \textcolor{green!65!black}{\C^3}) \\  &\textcolor{green!65!black}{\Omega^1\otimes (\C \oplus \C \oplus \C^3)}
\end{split} \& S_- \otimes (\textcolor{blue!65!black}{\C^2} \oplus \textcolor{green!65!black}{\C^2 \oplus \C^4}) \& \circled{\textcolor{blue!65!black}{\C}} \oplus \textcolor{green!65!black}{\C^3} \arrow[dl,"\id"] \\
\& \& \&  \circled{\textcolor{green!65!black}{\C}}
\end{tikzcd} \right].
\end{gather}

\section{The conormal bundle exact sequence}
The cotangent bundle, the conormal bundle, and the restriction of the cotangent bundle of the ambient $\P^7$ to the nilpotence variety sit in the conormal exact sequence, which is the dual of~\eqref{normalses}:
\begin{equation}
0 \longrightarrow \cN^\vee_Y \longrightarrow \Omega^1_{\P^7}|_Y \longrightarrow \Omega^1_Y \longrightarrow 0.
\end{equation}
In the same fashion as above, we now study the cohomology of these sheaves and their associated multiplets. 

\subsection{Cotangent bundle}
\subsubsection{Cohomology and Hilbert series}
As explained above, the cotangent bundle of the nilpotence variety $Y$ is given by the exterior sum 
\begin{equation}
\Omega^1_Y = \pi^\ast \Omega^1_{\P^1} \oplus \pi_3^* \Omega^1_{\P^3} = \Omega^1_{\P^1} \boxplus \Omega^1_{\P^3}, 
\end{equation}
where $\Omega^1_{\P^1} \cong \mathcal{O}_{\P^1} (-2).$
As a consequence, the associated multiplet is again a direct sum
\begin{equation}
\mu A^\bu(\Omega^1_Y) = \mu A^\bu(\cO_Y(-2,0)) \oplus \mu A^\bu(\pi_3^* \Omega^1_{\P^3}) \: .
\end{equation}
The multiplet $\mu A^\bu(\cO_Y(-2,0))$, arising from the cotangent bundle of $\P^1$, was already described in~\S\ref{sec: O(0,m)}, therefore we are left with describing $\mu A^\bu(\pi_3^* \Omega^1_{\P^3})$. To this end, we need to study the zeroth cohomology of  
\begin{equation}
\pi_3^*\Omega^1_{\P^3} (k) = \pi^\ast \cO_{\P^1}(k) \otimes \pi_3^*\Omega^1_{\P^3}(k) = \mathcal{O}_{\P^1}  (k) \boxtimes \Omega^1_{\P^3} (k).
\end{equation}
The K\"unneth theorem implies that 
\bear
H^0(\pi_3^* \Omega^1_{\P^3}(k)) = H^0(\cO_{\P^1}(k)) \otimes H^0( \Omega^1_{\P^3}(k)),
\eear
reducing the problem to compute the dimension of the zeroth cohomology of $\Omega^1_{\P^3} (k).$ This can be obtained by Bott formulas~\cite{Okonek} or by explicitly studying the twist of the dual
of the Euler exact sequence for $\Omega^1_{\proj 3}$,
\bear
\xymatrix{
	0 \ar[r] & \Omega^1_{\proj 3} (k) \ar[r] & \mathcal{O}_{\proj 3} (k-1) \otimes \mathbb{C}^{4} \ar[r] & \mathcal{O}_{\proj 3} (k) \ar[r] & 0.
}
\eear
In order to obtain a short exact sequence of modules, we have to check that the connection morphism vanishes. Clearly, if $k<0$ then $H^0 (\Omega^1_{\proj 3} (k)) = 0$. If $k=0$, then this corresponds to the Hodge number of $\mathbb{P}^3$ and in particular one finds $h^0 (\Omega^1_{\proj 3}) = 0 = h^{1,0} (\proj 3) $. For $k = 1$ it is easy seen that the map 
\bear
\varphi_{k=1} : H^0 (\mathcal{O}_{\proj 3}) \otimes \mathbb{C}^4 \rightarrow H^0 (\mathcal{O}_{\proj 1} (1)),
\eear
given by $ \mathbb{C}^{4} \owns (c_0, \ldots, c_3 ) \mapsto \sum_{i = 0}^3 c_i X_i $, for $\{ X_0, \ldots, X_1\}$ global sections of $\mathcal{O}_{\proj 3} (1)$ is an isomorphism and hence $H^0 (\Omega^1_{\proj 3} (1)) = 0.$ On the other hand in the case $k>1$ the map $\varphi_{k>1} : H^0 (\mathcal{O}_{\proj 3}(k-1)) \otimes \mathbb{C}^4 \rightarrow H^0 (\mathcal{O}_{\proj 1} (k))$ is only surjective so that $H^1 (\Omega^1_{\proj 3} (k)) = 0$ and $\ker (\varphi_{k>1}) = H^0 (\Omega^1_{\proj 3}(k))$. If follows that for $k >1$ one has 
\begin{align}
h^0 (\Omega^1_{\proj 3} (k))  = 4 \binom{k+2}{k-1} - \binom{k+3}{k}  
 = \frac{1}{2}(k+2)(k+1)(k-1).
\end{align}
In turn, this implies that
\bear
h^0(\pi_3^* \Omega^1_{\P^3}(k)) = \frac{1}{2}(k+1)^2(k+2)(k-1),
\eear	
and the related Hilbert series gives 
\begin{equation}
\Hilb(\pi_3^* \Omega^1_{\P^3}) = t^2 \frac{18 - 64t +89t^2 - 64t^3 + 28t^4 - 8t^5 +t^6}{(1-t)^8}.
\end{equation}

\subsubsection{Equivariant decomposition}
For the equivariant decomposition, we identify
\begin{equation}
H^0(\Omega^1_{\P^3}(k)) = [0|k-2,1,0] \: ,
\end{equation}
and find the following representations in $\mu A^\bu(\pi_3^*\Omega^1_{\P^3})$.
\begin{equation}
\begin{split}
W_0 &= \phantom{-}[2|0,1,0] \\
W_1 &= -[1|0,0,1] -[1|1,1,0] - [3|0,0,1] \\
W_2 &= \phantom{-}[0|0,0,0] + [0|0,2,0] + [0|1,0,1] + [2|0,0,0] + [2|1,0,1] + [4|0,0,0] \\
W_3 &= -[1|0,1,1] - [1|1,0,0] - [3|1,0,0] \\
W_4 &= \phantom{-}[0|0,0,2] + [2|0,1,0] \\
W_5 &= -[1|0,0,1] \\
W_6 &= \phantom{-}[0|0,0,0]
\end{split}
\end{equation}
The resulting multiplet takes the following form.
    \begin{gather}
\mu A^\bu(\pi_3^*\Omega^1_{\P^3})^\# =
\\
\left[
\begin{tikzcd}[row sep=0.3cm, column sep=0.12cm, ampersand replacement = \&]
    \Omega^1 \otimes \C^3 \& 
{\begin{array}{c}
\C^2 \otimes V \otimes S_+ \\  \C^4\otimes S_- 
\end{array}}
\& 
{\begin{array}{c}
\Sym^2(V) \oplus \Omega^2 \\ \C^3 \oplus \C^3 \otimes \Omega^2 \oplus \C^5 
\end{array}}
\& 
{\begin{array}{c}
\C^2 \otimes S_- \otimes V \\  S_+ \otimes \C^4
\end{array}} 
\&
{\begin{array}{c}
\Omega^3_- \\ \Omega^1 \otimes \C^3
\end{array}}
\& \C^2 \otimes S_- \& \C
\end{tikzcd} \right] \notag
\end{gather}

\subsection{Restriction of $\Omega^1_{\P^7}$ to the nilpotence variety}
\subsubsection{Cohomology and Hilbert series}
Dually to the case of $\mathcal{T}_{\proj 7}|_Y$, the cohomology of $\Omega^1_{\proj 7}|_Y$ and its twists is studied by restricting the dual of the Euler exact sequence for the ambient space $\proj 7$ to $Y$. This gives 
\bear
\xymatrix{
	0 \ar[r] & \Omega^1_{\proj 7}|_Y (k) \ar[r] & \mathcal{O}_{Y} (k-1, k-1) \otimes \mathbb{C}^8 \ar[r] & \mathcal{O}_{Y} (k,k)  \ar[r] & 0.
}
\eear
Studying the related long exact cohomology sequence, it is easy to see that if $k \leq 1$ then $H^0 (\Omega^1_{\proj 7}|_Y (k)) = 0.$
In the remaining case, when $k > 1$, the space of global sections $H^0 (\Omega^1_{\proj 7}|_Y  (k))$ is actually non-zero and the long cohomology sequence splits since the polynomial map 
\bear
\xymatrix{
	H^0 (\mathcal{O}_{Y}(k-1, k-1))\otimes \mathbb{C}^{8} \ar[rr]^{\qquad \quad (X_i, Y_j)} & &  H^0 (\mathcal{O}_Y (k,k))
}
\eear
is surjective, so one gets the short exact sequence
\begin{equation}
0 \rightarrow H^0 (\Omega^1_{\proj 7}|_Y  (k)) \rightarrow H^0 (\mathcal{O}_{Y} (k-1, k-1)) \otimes \mathbb{C}^8 \rightarrow H^0 (\mathcal{O}_{Y} (k,k))  \rightarrow 0.
\end{equation}
This says that 
\begin{equation}
\begin{split}
h^0 (\Omega^1_{\proj 7}|_Y  (k))  &= 8 k \binom{k+2}{k-1} - (k+1)\binom{k+3}{k}  \\
&= \frac{4}{3}k(k+2)(k+1)k - \frac{1}{6}(k+1)(k+3)(k+2)(k+1)  .
\end{split}
\end{equation}
Finally, for future use, notice that for $k> 1$ one has $H^1 (\Omega^1_{\proj 7}|_Y (k)) = 0,$ since it is mapped injectively into $H^1 (\mathcal{O}_Y(k-1, k-1)) \otimes \mathbb{C}^8$ which is indeed zero. \\
The related Hilbert series of $\Omega^1_{\proj 7}|_Y$ can be resummed easily to give
\begin{equation}
\Hilb(\Omega^1_{\P^7}|_Y) = t^2\frac{34 - 112t + 137t^2 - 80t^3 +28t^4 - 8t^5 +t^6}{(1-t)^8}
\end{equation}

\subsubsection{Equivariant decomposition}
In terms of representations, the sequence gives
\begin{equation}
H^0(\Omega^1_{\P^7}|_Y(k)) = [k-2|k,0,0] + [k|k-2,1,0] + [k-2|k-2,1,0] \: .
\end{equation}
Using these results, we can deduce the field content of $\mu A^\bu(\Omega^1_{\P^7}|_Y)$.

\begin{equation}
\begin{split}
W_0 &= \phantom{-2}[0|0,1,0] +\phantom{2}[0|2,0,0] + [2|0,1,0] \\
W_1 &= -2[1|0,0,1] - 2[1|1,1,0] - [3|0,0,1] \\
W_2 &= \phantom{-2}[0|0,0,0] + \phantom{2}[0|0,2,0] + [0|1,0,1] + 2[2|0,0,0] + 2[2|1,0,1] + [4|0,0,0] \\
W_3 &= -\phantom{2}[1|0,1,1] - \phantom{2}[1|1,0,0] - 2[3|1,0,0] \\
W_4 &= \phantom{-2}[0|0,0,2] + \phantom{2}[2|0,1,0] \\
W_5 &= -\phantom{2}[1|0,0,1] \\
W_6 &= \phantom{-2}[0|0,0,0]
\end{split}
\end{equation}
The resulting multiplet takes the following form.
\begin{gather}
\mu A^\bu(\Omega^1_{\P^7}|_Y)^\# = \\
\left[ \!\!\!
    \begin{tikzcd}[row sep=0.3cm, column sep=0.05cm, ampersand replacement = \&]
\begin{array}{c}
\Omega^1 \otimes(\C \oplus \C^3) \\ \Sym^2 S_+
\end{array} 
\& \begin{array}{c}
(\C^2 \otimes V \otimes S_+)^{\oplus 2} \\ \C^4 \otimes S_- 
\end{array}
\& \begin{array}{c}
\Sym^2(V) \oplus \Omega^2 \\ \Omega^2 \otimes \C^3 \oplus \C^5 \\ \Omega^2 \otimes \C^3  \oplus (\C^3)^{\oplus 2}
\end{array}
\& \begin{array}{c}
\C^2 \otimes V \otimes S_- \\ (S_+ \otimes \C^4)^{\oplus 2}
\end{array}
\& \begin{array}{c}
\Sym^2(S_-) \\  \Omega^1 \otimes \C^3
\end{array}
\&
S_- \otimes \C^2
\&
\C
\end{tikzcd} \!\!\! \right] \notag
\end{gather}

\subsection{The conormal bundle and its supergravity multiplet}
\subsubsection{Cohomology and Hilbert series}
\noindent Having available the cohomology of the cotangent bundle and the restriction of $\Omega^1_{\proj 7}$ to the nilpotence variety $Y$, one can study the conormal bundle and its related multiplet in a similar fashion as for the normal bundle above, \emph{i.e.}\ by considering $k$-twists of the conormal exact sequence~\eqref{conormalbundle}:
\bear \label{conormaltwist}
\xymatrix{
	0 \ar[r] & \mathcal{N}^\vee_{Y} (k) \ar[r] & {\Omega^1_{\proj 7}}\arrowvert_{Y}  (k) \ar[r] & \Omega^1_{Y} (k) \ar[r] & 0.
}
\eear
The issue one faces following this approach is that the related long exact cohomology sequence starts with a four-term sequence in the relevant case $k>1$ 
\begin{equation}
0 \rightarrow H^0 (\mathcal{N}^\vee_{Y} (k)) \rightarrow H^0 ({\Omega^1_{\proj 7}|_Y} (k) ) \rightarrow H^0 (\Omega^1_{Y} (k) ) \rightarrow H^1 (\mathcal{N}^\vee_{Y} (k)) \rightarrow 0,
\end{equation}
and it is not completely trivial to establish the vanishing of the group $H^1 (\mathcal{N}^\vee_{Y} (k))$ for any $k\geq 1.$ For the sake of exposition we have deferred this verification to Appendix~\ref{ap b}, which follows another route to describe the algebraic geometry of the normal and conormal bundle of $Y$, reducing the problem to studying easy polynomial maps. 

Using the above results, one computes
\begin{align}
h^0 (\mathcal{N}^\vee_{Y} (k))  = 8 k \binom{k+2}{3} - 2k\binom{k+3}{3} - (k^2-1)\binom{k+2}{2} 
 = \frac{1}{2}(k+1)(k-1)^2(k+2). 
\end{align}
The related Hilbert series is resummed to give
\begin{equation}
\Hilb(\cN^\vee_{Y}) = \sum_{k=2}^{\infty} \frac{1}{2}(k+1)(k-1)^2(k+2) t^k =
 t^2 \frac{6 - 8t - 17t^2 + 40t^3 - 28t^4 + 8t^5 - t^6}{(1-t)^8} \: .
\end{equation}

\subsubsection{Equivariant decomposition}
In terms of representations the conormal exact sequence implies
\begin{equation}
H^0(\cN^\vee_Y(k)) = [k-2|k-2,1,0] \: .
\end{equation}
We find the following field content for $\mu A^\bu(\cN^\vee_Y)$.
\begin{gather}
    W_0 = [0|0,1,0], \qquad
    W_1 = -[1|0,0,1], \qquad 
    W_2 = -[0|0,2,0] + [2|0,0,0], \\
    W_3 = [1|0,1,1], \qquad
    W_4 = -[0|0,0,2] - [2|0,1,0], \qquad
    W_5 = -[1|0,0,1], \qquad
    W_6 = [0|0,0,0]. \notag
\end{gather}
In summary, the multiplet takes the form
\begin{equation}
\mu A^\bu(\cN^\vee_Y)^\#  = \left[
\begin{tikzcd}[row sep=0.3cm, column sep=0.2cm]
V & S_- \otimes \C^2 & \C^{3} \\
& \Sym^2_0(V) & (V \otimes S_-)_{\frac{3}{2}} \otimes \C^2 & V \otimes \C^3 \oplus \Omega^3_- & S_- \otimes \C^2 & \C \: 
\end{tikzcd} \right].
\end{equation}
This is precisely the field content of six-dimensional $\cN=(1,0)$ supergravity, presented as the ``type-II Weyl multiplet''~\cite{LinchSugraProj}.

\subsection{Deformation}
There is again a deformation of $\mu A^\bu(\Omega^1_Y)^\# \oplus \mu A^\bu(\cN^\vee_Y)^\#$ such that the result is quasi-isomorphic to $\mu A^\bu(\Omega^1_{\P^7}|_Y)^\#$. Note that the cotangent bundle on $Y$ splits into the relevant factors on $\PP^1$ and $\PP^3$ (with the former giving rise to the multiplet $\mu A^\bu(0,2)$ that we discussed in~\S\ref{sec: O(0,m)}). To make the cancellations explicit, we add this splitting in our color-coding. 
\begin{gather}
\mu A^\bu(\Omega^1_{\P^7}|_Y)^\# \simeq
[\textcolor{yellow!65!black}{\mu A^\bu(\pi_1^* \Omega^1_{\PP^1})} \oplus \textcolor{green!65!black}{\mu A^\bu(\pi_3^*\Omega^1_{\PP^3})^\#} \oplus \textcolor{blue!65!black}{\mu A^\bu(\cN^\vee_Y)^\#}]^{\mathrm{Deform}} = \\
\left[
\begin{tikzcd}[row sep=0.3cm, column sep=0.1cm, ampersand replacement = \&]
\begin{split}
&\textcolor{green!65!black}{\Omega^1 \otimes \C^3} \\ &\textcolor{yellow!65!black}{\Sym^2 S_+} \\ &\textcolor{blue!65!black}{V}
\end{split} 
\& \begin{split}
&\textcolor{green!65!black}{\C^2 \otimes V \otimes S_+} \\ &\textcolor{green!65!black}{\C^4 \otimes S_-} \\ &\textcolor{yellow!65!black}{\C^2 \otimes (V \otimes S_+)_{\frac{3}{2}} } \\ & \textcolor{blue!65!black}{S_- \otimes \C^2}
\end{split}
\& \begin{split}
&\textcolor{green!65!black}{\Sym^2(V) \oplus \Omega^2} \\ &\textcolor{green!65!black}{\Omega^2 \otimes \C^3 \oplus \C^5} \\ &\circled{\textcolor{yellow!65!black}{\Sym^2_0(V)}} \\ &\textcolor{yellow!65!black}{\Omega^2 \otimes \C^3 } \\ &\textcolor{green!65!black}{\C^3} \oplus \textcolor{blue!65!black}{\C^3} \arrow[dl]
\end{split}
\& \begin{split}
&\textcolor{green!65!black}{\C^2 \otimes V \otimes S_-} \\ &\textcolor{green!65!black}{(S_+ \otimes \C^4)} \\ &\textcolor{yellow!65!black}{(S_+ \otimes \C^4)}\\ &\circled{\textcolor{yellow!65!black}{\C^2 \otimes (V \otimes S_-)_{\frac{3}{2}}}} \arrow[dl]
\end{split}
\& \begin{split}
&\circled{\textcolor{green!65!black}{\Omega^3_-}} \\ &\circled{\textcolor{green!65!black}{\Omega^1 \otimes \C^3}} \\ &\textcolor{yellow!65!black}{\Sym^2(S_-)} \\ &\textcolor{yellow!65!black}{V \otimes \CC^3} \arrow[dl]
\end{split}
\& \begin{split}
&\circled{\textcolor{green!65!black}{S_- \otimes \C^2}} \\ &\textcolor{green!65!black}{S_- \otimes \C^2} \arrow[dl]
\end{split}
\& \begin{split}
&\circled{\textcolor{green!65!black}{\C}} \\ &\textcolor{green!65!black}{\C} \arrow[dl]
\end{split} \\
\& \circled{\textcolor{blue!65!black}{\Sym^2_0(V)}} \&  \circled{\textcolor{blue!65!black}{\C^2 \otimes (V \otimes S_-)_{\frac{3}{2}}}} \&  \circled{\textcolor{blue!65!black}{\Omega^1 \otimes \C^3 \oplus \Omega^3_-}} \&  \circled{\textcolor{blue!65!black}{S_- \otimes \C^2}} \&  \circled{\textcolor{blue!65!black}{\C}}  
\end{tikzcd} \right] \notag 
\end{gather}

\appendix

\section{Equivariant vector bundles and representations} \label{ap a}

This paper features different methods from algebraic geometry and representation theory, whose mutual relations might appear somewhat puzzling on a first reading. In order to help the reader make his way through the manuscript, we now spell out the general philosophy connecting representations of a certain symmetry group or Lie algebra and the space of global sections of a certain related vector bundle.

For the sake of concreteness, we will work in the case which is relevant for this paper, that of complex projective spaces $\mathbb{P}^n$, but all of the following considerations hold true more generally on a generic homogeneous space $G/H$, for $G$ a complex Lie group and $H \subset G$ a closed subgroup of $G$.

We start by recalling that the complex projective space $\proj n$ can indeed be realized as a homogeneous space via the quotient (flag variety) 
$
SL (n+1, \mathbb{C}) / P,
$
where $P$ is the parabolic subgroup of block lower triangular matrices in $\SL(n+1, \mathbb{C})$, \emph{i.e.} $b \in P$ is of the form
\bear
b = \underbrace{\left ( \begin{array}{ccc|c}
		\ast & \ldots & \ast & 0   \\
		\vdots & \ddots & \vdots & \vdots   \\ 
		\ast & \ldots & \ast & 0   \\
		\hline
		\ast & \ldots & \ast & \ast  
	\end{array}
	\right )}_{n+1}.
\eear

In this context, an $\SL(n+1, \mathbb{C})$-equivariant vector bundle $\pi : E \rightarrow \proj n $ on $\proj n$ is a holomorphic vector bundle that carries a $\SL (n+1,\mathbb{C})$-action on its total space $E$ which is compatible with the $\SL (n+1, \mathbb{C})$-action on the base space $\proj n$. This means that we require the structure map $\pi : E \rightarrow \proj n$ to be $\SL(n+1, \mathbb{C})$-equivariant, \emph{i.e.} $\pi (g \cdot e) = g \cdot \pi (e) $ for all $e \in E$ and $g \in G$ and the translation between fibers to be linear, \emph{i.e} $\ell_g : E_{[x]} \rightarrow E_{g\cdot [x]}$ given by $e \mapsto g\cdot e$ is linear---thus an isomorphism of vector spaces for all $e \in E$ and $g \in G.$

It follows immediately from the given definition that the fiber $E_{P}$ at the point corresponding to the parabolic subgroup $P $ in $\mathbb{P}^n$ is a $P$-module, and the map $\rho : P \rightarrow \Aut (E_{P})$ given by $\rho (b) \defeq \ell_{b} : E_{P} \rightarrow E_{P} $ is a representation of $P$.
Conversely, starting from a holomorphic representation $\rho : P \rightarrow \Aut (V) $ for some $P$-module $V$, it is not hard to see that the associated vector bundle $\pi : \SL (n+1, \mathbb{C}) \times_{\rho} V \rightarrow \proj n $, with $\pi ([g, v]) \defeq gP$ is indeed a $\SL(n+1, \mathbb{C})$-equivariant vector bundle as defined above and that restricting the action of $P$ on the total space to the fiber $\{(1,v) : v \in V \}$ lying over $P \in SL (n+1, \mathbb{C}) / P \cong \mathbb{P}^n$, one recovers the original representation $\rho.$ Recall that here the total space $\SL (n+1, \mathbb{C}) \times_{\rho} V$ is the quotient manifold $(SL (n+1, \mathbb{C}) \times V) / \sim$ by the relations $(g, v) \sim (gb^{-1}, \rho (h) v) $ for $g\in \SL(n+1, \mathbb{C})$, $b \in P$, endowed with the quotient topology. These considerations suggest that there is a bijection 
\bear \xymatrix@R=1.5pt{
	\{ \mbox{holomorphic representations of } P \} \ar@{<->}[r] & \{ \SL(n+1, \mathbb{C})\mbox{-equivariant vector bundles on }  \mathbb{P}^n\} \nonumber \\
	\rho : P \rightarrow \Aut (V) \; \ar@{<->}[r] & \; \pi: \SL(n+1, \mathbb{C}) \times_\rho V \rightarrow \proj n. 
}
\eear
This is in fact an equivalence of categories and it can be used to {induce} representations from the parabolic subgroup $P$ to the whole group $\SL(n+1,\mathbb{C})$. The crucial point is that given an equivariant vector bundle $E_{\rho} \defeq SL (n+1, \mathbb{C}) \times_\rho V \stackrel{\pi}{\rightarrow} \proj n$ for a certain representation $\rho$ of $P$ and denoting with $\mathcal{E}_\rho$ the sheaf of its holomorphic sections, then the $\SL(n+1,\mathbb{C})$-action on the total space of $E_\rho$ induces a linear action on the (vector) space of global holomorphic sections $\Gamma (E_\rho) \defeq H^0 (\proj n, \mathcal{E}_\rho)$ given by
\bear \label{actionSL}
\xymatrix@R=1.5pt{
	SL (n+1, \mathbb{C}) \times H^0 (\proj n, \mathcal{E}_\rho) \ar[r] & H^0 (\proj n, \mathcal{E}_\rho)  \\
	(g, s) \ar@{|->}[r] & g\cdot s 
}
\eear
where $(g \cdot s) ([x]) \defeq g \cdot s (g^{-1}[x])$ for $g\in \SL(n+1, \mathbb{C})$, $s \in H^0 (\proj n, \mathcal{E}_\rho)$ and $[x] \in \proj n.$ More in general, by naturality and equivariance of $E_\rho$, one has a linear action of $\SL(n+1, \mathbb{C})$ on every cohomology group $R^q \, \Gamma (E_\rho)$ for $q = 0, \ldots, n$. 

Notice that the action \eqref{actionSL} does indeed defines a holomorphic representation 
\bear
\xymatrix{
	\varphi : \SL(n+1, \mathbb{C}) \ar[r] & \Aut (H^0 (\mathbb{P}^n, \mathcal{E}_\rho))
}\eear
of $\SL (n+1, \mathbb{C})$ on $H^0 (\proj n, \mathcal{E}_\rho)$, since $\proj n$ is compact and ${E}_\rho$ is a holomorphic vector bundle, and hence all its cohomology groups are finite-dimensional.
This construction is heavily exploited in the paper, where the space of global sections of a certain vector or line bundle on $\proj n$ will always be understood as carrying a certain representation of $\SL(n+1, \mathbb{C})$. Note that we take the liberty of moving freely from representation of $\SL (n+1, \mathbb{C})$ to representation of its Lie algebra $\mathfrak{sl} (n+1, \mathbb{C})$ and vice versa, since $\SL (n+1,\mathbb{C})$ is simply connected for all $n$ and hence there is a one-to-one correspondence between (finite-dimensional) representations of the Lie group and representations of its Lie algebra.

Now, recalling that $\Pic (\proj n) \cong \mathbb{Z}$, all of the $\SL(n+1, \mathbb{C})$-modules $H^0 (\proj n, \mathcal{O}_{\proj n} (k)) $ for $k \in \mathbb{Z}$ corresponding to global sections of line bundles on $\proj n$ are easily identified as
\bear
H^0 (\proj n, \mathcal{O}_{\proj n} (k)) \cong \left \{ \begin{array}{lr}
	\Sym^k (V_{n+1}^\vee) & k \geq 0\\
	0 & k <0,
\end{array}
\right.
\eear 
where $V_{n+1}^\vee \cong (\mathbb{C}^{n + 1})^{\vee} \cong \mbox{Span}_\mathbb{C}\{ x_0, \ldots, x_n\}$ is in an obvious fashion an irreducible $\SL (n+1, \mathbb{C})$-module. Notice that $\Sym^k (V_{n+1}^\vee) \cong \mathbb{C}[x_0, \ldots, x_n]_{(k)}$, the homogeneous polynomial of degree $k$. It follows that $H^0 (\proj n, \mathcal{O}_{\proj n} (k)) \cong \Sym^k (V_{n+1}^\vee)$ is as well an irreducible $\SL(n+1,\mathbb{C})$-module, corresponding to the irreducible representation of $A_n = \mathfrak{sl} (n+1, \mathbb{C})$ given by
\bear
H^0 (\proj n, \mathcal{O}_{\proj n} (k)) \; \longleftrightarrow \;  \underbrace{[k, 0, \ldots, 0]}_{n} 
\eear
in terms of its Dynkin labels for the representations of the $A_n$-series.
These findings are in agreement with Borel--Weil theorem, that adapted to the case of $\mathbb{P}^n \cong \SL(n+1, \mathbb{C}) / P$, constructs irreducible representations from the global sections of line bundles $\pi_{\lambda_k } : SL (n+1, \mathbb{C}) \times_{\lambda_k} \mathbb{C} \rightarrow \proj n $ associated to integral characters of $\lambda_k : P \rightarrow \mathbb{C}^{\times}$ of $P $ - which corresponds to irreducible holomorphic representations since $P$ is solvable. Calling $\mathcal{L}_{\lambda_n}$ the sheaf of holomorphic sections of the equivariant line bundle associated to $\lambda_k$ one can see that $\mathcal{L}_{\lambda_{k}} \cong \mathcal{O}_{\proj n} (k)$ so that $\mathcal{O}_{\proj n} (k)$ is indeed $\SL (n+1, \mathbb{C})$-equivariant and 
\bear
H^0 (\proj n, \mathcal{L}_{\lambda_k}) \cong H^0 (\proj n, \mathcal{O}_{\proj n} (k))
\eear
for any $k$. Indeed, the Borel--Weil theorem guarantees that $H^0 (\proj n, \mathcal{L}_{\lambda_k})$ is an irreducible $\SL (n+1, \mathbb{C})$-module of highest weight $k$ if and only $\lambda_k$ is a dominant weight: this is the case if and only if $k\geq0, $ that corresponds in fact to the irreducible representation $[k, 0, \ldots, 0]$.

Let us now pass to a higher-rank case which is relevant for the present paper, \emph{i.e.} that of the tangent bundle to the projective space $\proj n$ seen as the homogeneous space $\SL(n+1, \mathbb{C}) / P.$ Starting from the adjoint representation $\Ad_{\SL} : \SL (n+1, \mathbb{C}) \rightarrow \Aut (\mathfrak{sl} (n+1, \mathbb{C}))$ of $\SL (n+1,\mathbb{C})$, one can consider the restriction to $P \subset SL (n+1, \mathbb{C})$, given by $\Ad_{\SL} |_{P} : P \rightarrow \Aut (\mathfrak{sl} (n+1, \mathbb{C}))$. Denoting with $\mathfrak{p}$ the Lie algebra of $P$, it is easy to observe that
\bear
(\Ad_{\SL} |_{P} (p)) X = (\Ad_{\SL} (p)) X = (\Ad_{P}(p)) X \in \mathfrak{p},
\eear
for any $p \in P$ and $X \in \mathfrak{p}$, so that $\mathfrak{p}$ is an invariant subspace for the representation $\Ad_{\SL} |_{P}: P \rightarrow \Aut (\mathfrak{sl}(n+1, \mathbb{C}))$, one gets the factor representation 
$\Ad^{\perp}_{SL} : P \rightarrow \Aut (\mathfrak{sl} (n+1, \mathbb{C}) / \mathfrak{p})$, and the following is a short exact sequence of $P$-equivariant modules
\bear
\xymatrix{
	0 \ar[r] & \mathfrak{p} \ar[r] & \mathfrak{sl} (n+1, \mathbb{C}) \ar[r] & \mathfrak{sl} (n+1, \mathbb{C}) / \mathfrak{p} \ar[r] & 0.
}
\eear
The associated vector bundle on $\proj n$ of the factor representation $\Ad^{\perp}_{SL}$ is naturally isomorphic to the tangent bundle of $\proj n$, \emph{i.e.}\ we have 
\bear
\mathcal{T}_{\proj n} \cong \SL(n+1,\mathbb{C}) \times_{\Ad_{\SL}^{\perp}} \mathfrak{sl} (n+1, \mathbb{C}) / \mathfrak{p},
\eear
which makes $\mathcal{T}_{\proj n}$ into a $\SL(n+1, \mathbb{C})$-equivariant vector bundle on $\proj n$. We recall that $\mathcal{T}_{\proj n}$ is very ample in the sense of Hartshorne so that its vector space of global sections is non-vanishing and as a consequence it is an irreducible $\SL(n+1, \mathbb{C})$-module of dimension $(n+1)^2-1$, which is identified with the adjoint representation, \emph{i.e.} 
\bear
H^0 (\proj n, \mathcal{T}_{\proj n}) \; \longleftrightarrow \; \underbrace{[1, 0, \ldots, 0, 1]}_{n}
\eear
in the Dynkin label notation for the $A_n$-series.

\section{Geometry of the normal and conormal bundles of $Y$} \label{ap b}

In this appendix we study the geometry of the normal and conormal bundle of the nilpotence variety $Y \cong \proj 1 \times \proj 3$ in more detail. We will follow a different approach compared to the one in the main text, aiming at finding resolutions for these vector bundles in line bundles, which make possible to easily evaluate their cohomology.  

We start with the normal bundle, considering the following two short exact sequences:
\bear
\xymatrix{
	0 \ar[r] & \mathcal{O}_Y \oplus \mathcal{O}_Y \ar[r]^{i\oplus j \qquad  \quad } & \mathcal{O}_Y(1,0)^{\oplus 2} \oplus \mathcal{O}_{Y} (0,1)^{\oplus 4} \ar[r] & \mathcal{T}_{Y} \ar[r] &  0 }
\eear
\bear
\xymatrix{
	0 \ar[r] & \mathcal{O}_Y  \ar[r]^{k \qquad \;} & \mathcal{O}_Y(1,1)^{\oplus 8} \ar[r] & \mathcal{T}_{\proj 7 |_Y} \ar[r] &  0}.
\eear
The first one comes from the Euler exact sequences for $\proj 1$ and $\proj 3$ respectively. The second sequence is the restriction of the Euler exact sequence for the ambient variety $\proj 7$ to $Y$. These exact sequences fit in the following commutative diagram.
\bear \label{comm1}
\xymatrix{
	& & 0 \ar[d] & 0 \ar[d] \\
	0 \ar[r] & \ker (q) \ar[d]^{\cong}  \ar[r]^{\iota} & \mathcal{O}_Y \oplus \mathcal{O}_Y \ar[d]^{i \oplus j} \ar[r]^q & \mathcal{O}_Y \ar[d]^k \\
	& \ker (f)\ar[d] \ar[r] & \mathcal{O}_Y (1, 0)^{\oplus 2} \oplus \mathcal{O}_{Y} (0,1)^{\oplus 4} \ar[r]^{\qquad f} \ar[d] & \mathcal{O}_{Y} (1,1)^{\oplus 8} \ar[d]^{E|_Y} \ar[r] & \ar[d]Q \\
	& 0\ar[r] & \mathcal{T}_Y \ar[d] \ar[r]  & \mathcal{T}_{\proj 7|_Y} \ar[d] \ar[r] & \mathcal{N}_{Y} \ar[r] & 0 \\
	& & 0 & 0}
\eear
Working on the affine cones of $C_a\proj 1 \times C_a \proj 3$ with coordinates $(x_0 , x_1)$ and $(y_0, y_1, y_2, y_3 )$ the above maps are defined as follows. The map $i \oplus j : \mathcal{O}_Y \oplus \mathcal{O}_Y \rightarrow \mathcal{O}_{Y} (1, 0)^{\oplus 2} \oplus \mathcal{O}_{Y} (0,1)^{\oplus 4}$ is the multiplication by the following $2 \times 6$ matrix
\bear
[i\oplus j] = \left ( \begin{array}{cc}
	x_0 & 0 \\
	x_1 & 0 \\
	0 & y_0 \\
	0 & y_1 \\
	0 & y_2 \\
	0 & y_3 
\end{array}
\right )
\eear
where frames in $\mathcal{O}_Y \oplus \mathcal{O}_Y$ are seen as column vectors. The map $k : \mathcal{O}_{Y} \rightarrow \mathcal{O}_{Y} (1,1)^{\oplus 8}$ is the multiplication by the following $1\times 8$ matrix
\bear
[k] = \left ( \begin{array}{c}
	x_0 y_0 \\
	x_0 y_1 \\
	\vdots  \\
	x_1 y_2 \\
	x_1 y_3 
\end{array}
\right )
\eear
The crucial map in the above diagram is the map $f : \mathcal{O}_Y (1,0)^{\oplus 2} \oplus \mathcal{O}_Y (0,1)^{\oplus 4} \rightarrow \mathcal{O}_{Y}(1,1)^{\oplus 8}$ which is given by the Jacobian of the Segre embedding. In coordinates it is described by the $8 \times 6$ matrix
\bear
[f] = \left ( \begin{array}{cccccc}
	y_0 & 0 &  x_0 & 0 & 0 & 0 \\
	y_1 & 0 &  0 & x_0 & 0 & 0 \\
	y_2 & 0 &  0 & 0 & x_0 & 0\\
	y_3 & 0 &  0 & 0 & 0 & x_0 \\
	0 & y_0 &  x_1 & 0 & 0 & 0\\
	0 & y_1 & 0 & x_1 & 0 & 0\\
	0 & y_2 &  0 & 0 & x_1 & 0\\
	0 & y_3 &  0 & 0 & 0 & x_1
\end{array}
\right )
\eear
It is easy to see that from the above expressions that $f \circ (i \oplus j) = k \oplus k$, hence by commutativity, one has $f \circ (i \oplus j ) = k \circ q = k\oplus k$. This says that the map $q : \mathcal{O}_Y \oplus \mathcal{O}_Y \rightarrow \mathcal{O}_Y$ corresponds to the sum of the elements in $\mathcal{O}_Y \oplus \mathcal{O}_Y$. In turn, this implies that $\ker (q) \cong \mathcal{O}_Y$ and the map $\iota : \mathcal{O}_Y \rightarrow \mathcal{O}_Y \oplus \mathcal{O}_Y$ is simply given by $f\mapsto (f, -f)$. On the other hand, it is clear that $\ker (q) \cong \ker (f)$, so that also $\ker (f) \cong \mathcal{O}_Y.$ The above diagram \eqref{comm1} thus reads
\bear
\xymatrix{
	& 0\ar[d] & 0 \ar[d] & 0 \ar[d] \\
	0 \ar[r] & \mathcal{O}_Y \ar[d]^{\cong}  \ar[r]^{\iota} & \mathcal{O}_Y \oplus \mathcal{O}_Y \ar[d]^{i \oplus j} \ar[r]^q & \mathcal{O}_Y \ar[d]^k \ar[r] & \ar[d] 0\\
	0 \ar[r] & \mathcal{O}_Y\ar[d] \ar[r] & \mathcal{O}_Y (1, 0)^{\oplus 2} \oplus \mathcal{O}_{Y} (0,1)^{\oplus 4} \ar[r]^{\qquad f} \ar[d]^{E_{xy}} & \mathcal{O}_{Y} (1,1)^{\oplus 8} \ar[d]^{E|_Y} \ar[r] & \ar[d]^\cong Q \ar[r]  & 0\\
	& 0\ar[r] & \mathcal{T}_Y \ar[d] \ar[r]^{d}  & \mathcal{T}_{\proj 7|_Y}  \ar[d] \ar[r] & \mathcal{N}_{Y} \ar[d] \ar[r] & 0 \\
	& & 0 & 0 & 0}
\eear
This is enough to read out a resolution of the normal bundle. Namely, we have
\bear \label{resnormal}
\xymatrix{
	0 \ar[r] & \mathcal{O}_Y \ar[r]^{\beta \qquad \qquad \quad} \ar[r] & \mathcal{O}_Y(1,0)^{\oplus 2} \oplus \mathcal{O}_Y (0,1)^{\oplus 4} \ar[r]^{\qquad \quad f} & \mathcal{O}_{Y} (1,1)^{\oplus 8} \ar[r]^{\quad \epsilon } & \mathcal{N}_{Y} \ar[r] & 0,
}
\eear
where $\beta \defeq \iota \circ (i \oplus j)$ is injective as it is the composition of two injective morphisms. Also, notice that it can be explicitly verified that $f \circ \beta = 0.$ Finally, $\epsilon$ is again the composition of two surjective morphisms, and hence it is surjective. Also, it is easy to see that $\epsilon \circ f = 0, $ since, with reference to the above commutative diagram, one has $d \circ E_{xy} = E |_Y \circ f: $ composing with the map $\mathcal{T}_{\proj 7|_Y} \rightarrow \mathcal{N}_{Y}$ gives zero by exactness of the normal bundle short exact sequence. Now, in order to study the cohomology of the normal bundle, together with all its twists, we break the previous resolution in two short exact sequences
\bear \label{sesM1}
\xymatrix{
	0 \ar[r] & \mathcal{O}_Y \ar[r] & \mathcal{O}_Y (1,0)^{\oplus 2} \oplus \mathcal{O}_Y (0,1)^{\oplus 4} \ar[r] & \mathcal{M} \ar[r] & 0,
}
\eear
\bear \label{sesM2}
\xymatrix{
	0 \ar[r] & \mathcal{M} \ar[r] & \mathcal{O}_Y (1,1)^{\oplus 8 } \ar[r] & \mathcal{N}_{Y} \ar[r] & 0,
}
\eear
which implicitly define the sheaf $\mathcal{M}$. The first one can be used to compute the cohomology of $\mathcal{M}$: in particular one finds that $H^1 (\mathcal{M} (k)) = 0$ for any $k$, so that the long exact cohomology sequence of the short exact sequence \eqref{sesM2} gives  
\bear
\xymatrix{
	0 \ar[r] & H^0 (\mathcal{M} (k)) \ar[r] & H^0 (\mathcal{O}_Y (k+1,k+1)^{\oplus 8 } ) \ar[r] & H^0 (\mathcal{N}_{Y} (k) )  \ar[r] & 0.
}
\eear
Computing $H^0(\mathcal{M} (k))$ from \eqref{sesM1}, one gets 
\begin{align}
h^0 (\mathcal{N}_{Y}) & = \frac{1}{2}(k+3)^2(k+2)(k+5)
\end{align}

\noindent We now consider the conormal bundle. First it can be observed that the \eqref{resnormal} is a resolution via vector bundles. This means that every map has constant rank and dualizing gives again an exact sequence
\begin{equation}\label{resconormal}
	0 \rightarrow \mathcal{N}^\vee_{Y} \xrightarrow{\epsilon^\vee}  \mathcal{O}_{Y} (-1,-1)^{\oplus 8} \xrightarrow{f^\vee}  \mathcal{O}_Y(-1,0)^{\oplus 2} \oplus \mathcal{O}_Y (0,-1)^{\oplus 4} \xrightarrow{\beta^\vee}  \mathcal{O}_Y \rightarrow  0.
\end{equation}

Once again the \eqref{resconormal} can be broken into two short exact sequences, namely
\bear \label{sesconorm1}
\xymatrix{
	0 \ar[r] & \mathcal{N}^\vee_{Y} \ar[r] & \mathcal{O}_{Y} (-1,-1)^{\oplus 8} \ar[r] & \mathcal{M}^\vee \ar[r] & 0,
}
\eear
\bear \label{sesconorm2}
\xymatrix{
	0 \ar[r] & \mathcal{M}^\vee \ar[r] & \mathcal{O}_Y(-1,0)^{\oplus 2} \oplus \mathcal{O}_Y (0,-1)^{\oplus 4} \ar[r] & \mathcal{O}_Y \ar[r] & 0.
}
\eear
We first need to study $\mathcal{M}^\vee$. We have the following commutative diagram.
\bear
\xymatrix{
	& 0 \ar[d] & 0 \ar[d] & 0 \ar[d]  \\
	0 \ar[r] & \mathcal{N}^\vee_{Y} \ar[r] \ar[d]^{=} & \Omega^1_{\proj 7|_Y} \ar[d] \ar[r] & \Omega^1_{Y} \ar[d] \ar[r] & 0 \\
	0 \ar[r] & \mathcal{N}^\vee_{Y} \ar[d] \ar[r] & \mathcal{O}_{Y} (-1,-1) \ar[d] \ar[r] & \mathcal{M}^\vee \ar[d] \ar[r] & 0 \\
	& 0 \ar[r] & \mathcal{O}_Y \ar[d] \ar[r]^{\cong} & Q \ar[r] \ar[d] & 0\\
	& & 0 & 0
}
\eear
Notice that $Q$ is of rank 1, hence $Q\cong \mathcal{O}_Y$, which shows that $\mathcal{M}^\vee$ is a (non split) extension 
\bear
\xymatrix{
	0 \ar[r] & \Omega^1_Y \ar[r] & \mathcal{M}^\vee \ar[r]  & \mathcal{O}_Y \ar[r] & 0.
}
\eear
Studying the cohomology of $\mathcal{M}^\vee (k)$ for $k\geq 1$---when $\mathcal{N}^\vee_{Y}$ can have global sections---one finds $H^1 (\mathcal{M}^{\vee} (k)) = 0$ and
\begin{align}
h^0 (\mathcal{M}^\vee (k)) & = h^0 (\Omega^1_Y (k)) + h^0 (\mathcal{O}_Y (k, k)) 
= 2k \binom{k+3}{3} + (k^2-1) \binom{k+2}{2}.
\end{align}
Going back to the conormal bundle $\mathcal{N}^\vee_{Y}$ and looking at the resolution twisted by $\mathcal{O}_Y(k)$ and split into the above two exact sequences \eqref{sesconorm1} and \eqref{sesconorm2} one has
\begin{align} \hspace*{-2cm} 
\xymatrix@R=14pt{
	0 \ar[r] & \mathcal{N}^\vee_{Y} (k) \ar[r]^{\pi^\vee \qquad } & 
        \mathcal{O}_{Y} (k-1,k-1)^{\oplus 8} \ar@/^2pc/[rr]^{f^\vee} \ar[r]^{\qquad \quad \bar f^\vee} & \mathcal{M}^\vee(k) \ar[dr] \ar[r]^{ i \quad \qquad \qquad } &  \mathcal{O}_Y(k-1,k)^{\oplus 2} \oplus \mathcal{O}_Y (k,k-1)^{\oplus 4}
        \ar[r]^{\qquad \quad \qquad \beta^\vee} & \mathcal{O}_Y (k,k)\ar[r] & 0\\
	& &  0 \ar[ur] & & 0
}
\end{align}
where $\bar f^\vee$ is surjective and $i$ is injective.
Then, the long exact cohomology sequence of~\eqref{sesconorm1} yields
\begin{equation}
	0 \rightarrow   H^0 (\mathcal{N}^\vee_{Y}) \rightarrow H^0 (\mathcal{O}_{Y} (k-1,k-1)^{\oplus 8} ) \xrightarrow{ \bar{f}^\vee}  H^0 (\mathcal{M}^\vee(k)) \rightarrow H^1 (\mathcal{N}^\vee_{Y}) \rightarrow 0.
\end{equation}
Hence in particular, 
\bear
H^1 (\mathcal{N}^\vee_{Y} (k)) = \mbox{coker} (H^0 (\mathcal{O}_{Y} (k-1,k-1)^{\oplus 8} ) \stackrel{\bar{f}^\vee}{\longrightarrow} H^0 (\mathcal{M}^\vee(k))),
\eear
so that
\begin{align} 
h^1 (\mathcal{N}^\vee_{Y} (k)) &  = h^0 (\mathcal{M}^\vee (k)) - 8 h^0 (\mathcal{O}_Y(k-1, k-1)) + \dim ( \ker \bar f^\ast ).
\end{align}
Now, one has $f^\vee = i \circ \bar f^\vee$, but since $i $ is injective, then $\ker (f^\vee) = \ker (\bar f^\vee)$, where 
\bear
f^\vee :  \mathcal{O}_{Y} (k-1,k-1)^{\oplus 8}  \rightarrow \mathcal{O}_Y(k-1,k)^{\oplus 2} \oplus \mathcal{O}_Y (k,k-1)^{\oplus 4}. 
\eear
This implies that  
\begin{align} \label{h1}
h^1 (\mathcal{N}^\vee_{Y} (k)) & = h^0 (\mathcal{M}^\vee (k)) - 8 h^0 (\mathcal{O}_Y(k-1, k-1))  + \dim ( \ker (f^\vee )) \nonumber \\ 
& = 2k \binom{k+3}{3} + (k^2-1) \binom{k+2}{2} - 8k \binom{k+2}{3} + \dim (\ker (f^\vee))
\end{align}
Notice that $\dim (\ker (f^\vee) ) = h^0 (\mathcal{N}^\vee_{Y})$, hence the kernel of the maps yields the number of global sections of the conormal bundle. The map $f^\vee$ can indeed be described explicitly, and it is given by the transpose of $f$ introduced above for the normal bundle: it is a $6\times 8$ matrix given in the above affine cones coordinates 
\bear
[f^\vee] = [f]^\top = \left ( 
\begin{array}{cccccccc}
	x_0 & x_1 & x_2 & x_3 & 0 & 0 & 0 & 0 \\
	0 & 0 & 0 & 0 & x_0 & x_1 & x_2 & x_3  \\
	u_0 & 0 & 0 & 0 & u_1 & 0 & 0 & 0 \\
	0 & u_0 & 0 & 0 & 0 & u_1 & 0 & 0 \\
	0 & 0 & u_0 & 0 & 0 & 0 & u_1 & 0 \\
	0 & 0 & 0 & u_0 & 0 & 0 & 0 & u_1 
\end{array}
\right ).
\eear
In the special case $k=1$ the long exact cohomology sequence reads
\begin{equation}
	0 \rightarrow H^0 (\mathcal{N}^\vee_{Y} (1)) \rightarrow H^0 (\mathcal{O}_Y^{\oplus 8}) \rightarrow H^0 (\mathcal{M}^\vee (1)) \cong H^0 (\mathcal{O}_Y (1,1)) \rightarrow H^1 (\mathcal{N}^\vee_{Y} (1) )\rightarrow 0.
\end{equation}
Here the middle map is an isomorphism, corresponding to the multiplication by $u_i x_j$. It follows that $H^0 (\mathcal{N}^\vee_{Y} (1)) \cong 0 \cong  H^1 (\mathcal{N}^\vee_{Y} (1)).$
In the remaining case, for $k\geq 2$, the kernel of the map $f^\vee$ in cohomology can be evaluated directly, and it yields 
\bear
h^0 (\mathcal{N}^\vee_{Y} (2)) = \ker (f^\vee) = 8k \binom{k+2}{3} - 2k \binom{k+3}{3} - (k^2-1) \binom{k+2}{2}.
\eear 
so that in particular it follows from \eqref{h1} that $H^1 (\mathcal{N}^\vee_{Y}) = 0.$

\printbibliography
\end{document}